\documentclass[10pt]{article}
\usepackage{graphicx, amsmath, amsthm, amsfonts, enumerate, authblk, amssymb, mathrsfs}

\newtheorem{theorem}{Theorem}[section]
\newtheorem{corollary}[theorem]{Corollary}
\newtheorem{lemma}[theorem]{Lemma}
\newtheorem{proposition}[theorem]{Proposition}
\newtheorem{remark}[theorem]{Remark}
\newtheorem{definition}[theorem]{Definition}

\def\dt{\partial^{(t)}}
\def\dnt{\partial^{(t)}_n}
\def\dtau{\partial^{(\tau)}}
\def\dntau{\partial^{(\tau)}_n}

\def\e{\varepsilon}

\DeclareMathOperator*{\dom}{\mathrm{dom}}

\makeatletter
\def\blfootnote{\xdef\@thefnmark{}\@footnotetext}
\makeatother

\begin{document}

\title{\sc Norm-resolvent convergence of one-dimensional high-contrast periodic problems to a Kronig-Penney dipole-type model}
\author[1]{Kirill\,D.\,Cherednichenko}
\author[2]{Alexander\,V.\,Kiselev}
\affil[1]{Department of Mathematical Sciences, University of Bath, Claverton Down, Bath, BA2 7AY, United Kingdom}
\affil[2]{Dragomanov National Pedagogical University, 9 Pyrohova St, Kyiv, 01601, Ukraine}

\maketitle

\par{\raggedleft\slshape To the memory of Professor Yuri Safarov\par}

\begin{abstract}
We prove operator-norm resolvent convergence estimates for one-dimensional periodic differential operators with rapidly oscillating coefficients
in the non-uniformly elliptic high-contrast setting, which has been out of reach of the existing homogenisation techniques.  Our asymptotic analysis is based on a special representation of the resolvent of the operator in terms of the $M$-matrix of an associated boundary triple (``Krein resolvent formula'').
  The resulting asymptotic behaviour is shown to be described, up to a unitary equivalent transformation, by a non-standard version of the Kronig-Penney model on $\mathbb R$.
\end{abstract}



\section{Introduction}
\label{intro_section}
It has been exploited in the mathematical analysis of periodic composite media, see {\it e.g.}  \cite{BLP}, \cite{BP},  \cite{JKO}, that they are amenable to the asymptotic analysis with respect to the period of the composite.
The related techniques, forming part of the mathematical theory of homogenisation,
are concerned with the asymptotic behaviour
of families of
operators associated with boundary-value
problems for differential equations with periodic
coefficients:
\begin{equation}
-{\rm div}\bigl(A^\varepsilon(x/\varepsilon)\nabla u\bigr)-zu=f,\ \ \ \ f\in L^2({\mathbb R}^d),\ \ \ \ \ \varepsilon>0,\ \ \ z<0,
\label{generic_eq}
\end{equation}
where for all $\varepsilon>0$ the matrix $A^\varepsilon$ is $Q$-periodic, $Q:=[0,1)^d,$ and may additionally be required to satisfy the
condition of uniform ellipticity:
\begin{equation}
A^\varepsilon(y)\ge\nu I,\ \ \ y\in Q,
\label{un_el}
\end{equation}
where $\nu>0$ is the ellipticity constant. The aim of these techniques is to describe an ``effective medium'', which represents the
family (\ref{generic_eq}) in the limit of vanishing ``microstructure size'' $\varepsilon,$ so that the corresponding ``limit'' equation, as $\varepsilon\to 0,$
has the form
\begin{equation}
-{\rm div}\bigl(A_{\rm hom}\nabla u\bigr)-zu=f,
\label{effective}
\end{equation}
with a constant matrix $A_{\rm hom}>0.$

A relatively recent area of interest within homogenisation is the behaviour of periodic media with
``high contrast'',
see {\it e.g.} \cite{Zhikov2000}, \cite{KamSm2}, \cite{CC}, where the smallest eigenvalue of the
matrix $A^\varepsilon>0$ in (\ref{generic_eq}) goes to zero as $\varepsilon\to0,$ {\it i.e.} the condition (\ref{un_el})
no longer holds and hence the differential operators in (\ref{generic_eq}) are not uniformly elliptic.
High-contrast composites play a key part in modelling photonic band-gap materials (see {\it e.g.} \cite{KamSm1}, \cite{CKS}) and media with
negative material properties (see {\it e.g.} \cite{BouchitteFelbacq}, \cite{Kohn_Shipman}).

In addition to their practical importance in modelling advanced materials, high-contrast composites are a source of new analytical challenges compared to the ``classical'' moderate-contrast
materials described by (\ref{generic_eq}). It has been well understood that the effective parameters $A_{\rm hom}$ in  (\ref{effective}) are given by the leading-order term at the zero energy $\lambda=0$ of the energy-quasimomentum dispersion relation
$\lambda_1=\lambda_1^\varepsilon(
\varkappa
)=
A_{\rm hom}\varkappa\cdot\varkappa
+O(\varkappa^3),$ $\varkappa\to0,$ for the first eigenvalue in the problem
\begin{equation}
-\bigl(\nabla+{\rm i}
\varkappa\bigr)\cdot A^\varepsilon\bigl(\nabla+{\rm i}
\varkappa
\bigr)u=\lambda u,\ \ \ \ u\in L^2(Q),\ \ \ \varkappa\in\bigl[0,2\pi\bigr)^d,
\label{fibre_equation}
\end{equation}
with respect to the scaled variable $y=x/\varepsilon\in Q,$ so that $A^\varepsilon=A^\varepsilon(y),$ and the gradient $\nabla$ in (\ref{fibre_equation}) is taken with respect to
$y.$   The link between the effective properties of the operator in (\ref{generic_eq}) and the asymptotics of $\lambda_1^\varepsilon(\varkappa)$  was first studied in  \cite{BLP} for elliptic and  \cite{Zhikov1989} for parabolic equations. The direct fibre decomposition into problems (\ref{fibre_equation}), followed by a perturbation analysis of its eigenvalue
$\lambda_1^\varepsilon(\varkappa)$
in each fibre,
allows one to obtain sharp operator-norm resolvent convergence estimates for the problem (\ref{generic_eq}), see \cite{Zhikov1989}, \cite{BirmanSuslina}. The related asymptotic results can be interpreted as a ``threshold effect near $\lambda=0$'' (see \cite{BirmanSuslina}, who coined the term in the context of  homogenisation) for the resolvent of the operator
$-\nabla\cdot A^\varepsilon\nabla$ in
$L^2({\mathbb R}^n),$ due to the relation
\begin{equation}
\bigl\{-{\rm div}_x\bigl(A^\varepsilon(x/\varepsilon)\nabla_xu\bigr)-z\bigr\}^{-1}=\varepsilon^2\bigl\{-{\rm div}_y\bigl(A^\varepsilon(y)\nabla_yu\bigr)-\varepsilon^2z\bigr\}^{-1},
\label{resolvents_scaling}
\end{equation}
so that the rescaled spectral parameter $\varepsilon^2z$ goes to zero as $\varepsilon\to0$ for a fixed $z.$  However, in order for this approach to work in the case of general coefficient matrices $A^\varepsilon,$ it is crucial that
the sequence $\{\lambda_2^\varepsilon(\varkappa)\}_{\varepsilon>0}$ be separated from zero uniformly in $\varepsilon$ and $\varkappa.$
Here $\{\lambda_j^\varepsilon(\varkappa)\}_{j=1}^\infty$ is the sequence of
all eigenvalues of (\ref{fibre_equation})--(\ref{generic_high_contrast}) for each $\varepsilon, \varkappa,$ indexed by $j$ in
non-decreasing order.
This condition is not satisfied for periodic models of ``double-porosity'', whose typical representative is described by
\begin{equation}
A^\varepsilon(y)=\left\{\begin{array}{ll}1, \ \ \ \ y\in Q_1,\\[0.15em] \varepsilon^2, \ \ \ y\in Q_0,\end{array}\right.
\label{generic_high_contrast}
\end{equation}
where $Q_0\cup Q_1=Q$ and $Q_0\neq\emptyset$ satisfies some minimal smoothness requirements.
It is easily seen that in this case  $\lambda_j^\varepsilon(\varkappa)\to0$ as $\varepsilon\to0,$ for all
$\varkappa\in[0,2\pi)^d,$ $j\in{\mathbb N}.$
Additional non-trivial analysis shows that for $l=0, 2,$ there are infinitely many
functions $j:(0,1]\to{\mathbb N}$
such that $\varepsilon^{-l}\lambda_{j(\varepsilon)}^\varepsilon(\varkappa)$ is  continuous in $\varepsilon,$ $\varkappa,$ and
tends to a finite non-zero limit as $\varepsilon\to0.$

This implies, in particular, that no equation of the form (\ref{effective}) describes the behaviour of (\ref{generic_eq}), (\ref{generic_high_contrast})
correctly in the resolvent sense, {\it i.e.} with an operator-norm smallness  estimate for the difference between the resolvent of
(\ref{generic_eq}), (\ref{generic_high_contrast}) and the resolvent of (\ref{effective}). These observations necessitate the development of analytical tools capable of dealing with the high-contrast problem (\ref{generic_eq}), (\ref{generic_high_contrast}).
In our approach, which we develop in the present paper for the one-dimensional situation, the operator on a fibre is considered as an extension of a suitably chosen ``minimal'' closed symmetric operator with equal finite deficiency indices. The extension theory, rooted in the classical works of J. von Neumann  \cite{JvN} and its further development by  M. G. Kre\u\i n \cite{Krein}, M. I. Vi\v sik \cite{Vishik}, M. S. Birman \cite{Birman} (commonly known as the Birman-Kre\u\i n-Vi\v sik theory), was reformulated in abstract terms in \cite{Ko1,Gor,DM} as the theory of boundary triples (see a brief exposition below, Section \ref{triples}). It relies on an abstract Green formula, which expresses the sesquilinear form of a maximal symmetric operator in terms of two boundary operators from the original Hilbert space to a ``boundary space''. In our setting the boundary space is finite-dimensional, hence the basic version of the theory is applicable, whereby both boundary operators are assumed to be surjective, and the (self-adjoint) extension under consideration is parameterised by a Hermitian matrix, exactly as in the Birman-Kre\u\i n-Vi\v sik approach. The
main analytic tool in the study of (proper) extensions of the minimal operator is
then
the
Weyl-Titchmarsh $M$-function, which is a generalisation of the classical Weyl-Titchmarsh $m$-coefficient, see {\it e.g.} \cite{Titchmarsh}. We remark that the $M$-function often plays a central r\^ ole in the spectral analysis of partial differential equations (PDE), where it is usually referred to as the Dirichlet-to-Neumann map. The advantages of using the above abstract approach are twofold: firstly, in this way the spectral analysis of the original problem can be reduced to the analysis of finite-dimensional matrices that depend analytically on the spectral parameter, and secondly, the celebrated Kre\u\i n formula (see Section \ref{sect:Krein}), expressing the (generalised) resolvent of the operator extension considered in terms of the resolvent of a given proper self-adjoint extension $A_\infty,$ allows one to use the Glazman splitting procedure \cite{Glazman}, where $A_\infty$ is a suitable ``split
operator''.

Our main result is the asymptotics, in the norm-resolvent sense, of a sequence of differential operators with periodic rapidly oscillating coefficients with high contrast:
\begin{equation}
-\bigl(a^\varepsilon\bigl(x/\varepsilon\bigr)u'\bigr)'-zu=f,\ \ \ \ f\in L^2({\mathbb R}),\ \ \ \varepsilon>0,\ \ \ z\in{\mathbb C},
\label{orig_problem}
\end{equation}
where, for all $\varepsilon>0,$ the coefficient $a^\varepsilon$ is $1$-periodic and
\begin{equation}
a^\varepsilon(y):=\left\{\begin{array}{lll}a_1,\ \ \ y\in[0, l_1),\\[0.2em] \varepsilon^2,\ \ \ y\in[l_1, l_1+l_2),\\ a_3,\ \ \ y\in[l_1+l_2, 1),\end{array}\right.
\label{high_contrast_case}
\end{equation}
with $a_1,$ $a_3>0,$ and $0<l_1<l_1+l_2<1.$ In a physical context ({\it e.g.} elasticity, porous-medium flow, electromagnetism) this represents a laminar composite medium of the double-porosity type, with $[0, l_1)$ and $[l_1+l_2, 1)$ referred to as the ``stiff'' components and  $[l_1, l_1+l_2)$
as the ``soft'' component of the composite (in terms of the ``unit cell'' $[0,1)$).  It has been noticed in \cite{Zhikov2005} that the spectra of a class of multi-dimensional versions of (\ref{orig_problem})
have the remarkable property of an infinite set of gaps opening in the limit of a vanishing period. The corresponding fact for laminar high-contrast media (equivalently, one-dimensional operators with high contrast) does not follow from the analysis of \cite{Zhikov2005}, and was established separately in \cite{CherednichenkoCooper}.
However, neither work goes as far as to establish the behaviour as $\varepsilon\to0$ of the resolvents of the $\varepsilon$-dependent operators describing the heterogeneous medium, in the operator-norm sense. As is argued by \cite{CC} in the multi-dimensional case, the resolvent asymptotics is not recovered by the standard two-scale analysis and requires a uniform asymptotic analysis of all components in the fibre decomposition of the underlying periodic operator. In the present work we utilise a version of the
Kre\u\i n formula,
written for a suitable boundary triple, in order to provide such a uniform asymptotics for (\ref{orig_problem}).


We start by providing auxiliary material leading up to a representation of the resolvents of (\ref{orig_problem}) in terms of a
family of resolvents of the elements of their fibre decompositions $A_\e^{(t)},$ $t\in[0,2\pi\varepsilon^{-1}).$ We develop a new approach to the analysis of this family, by
considering it as defined on a particular finite compact metric graph, thus bridging a gap between the problem of homogenisation of the family (\ref{orig_problem})--(\ref{high_contrast_case}) and the seemingly unrelated subject of spectral analysis of quantum graphs (see {\it e.g.} \cite{Kuchment2} and references therein). This includes (Section \ref{preliminaries1}) a description of the Gelfand transform, the boundary triple, and the Green formula associated with (\ref{orig_problem}), as well as a derivation of the corresponding $M$-matrix and a discussion of its invertibility properties. We also carry out (Section \ref{preliminary}) a useful rescaling of the problem on the fibre, and recall the Kre\u\i n resolvent formula, which is key to the analysis of the subsequent sections.

In Section \ref{intermediate_comparison} we show that the resolvents of the operators $A_\e^{(t)},$ $t\in[0,2\pi\varepsilon^{-1}),$ in the fibre decomposition of (\ref{orig_problem}) are close, in the operator-norm sense, to the family of generalised resolvents $\bigl(\tilde{A}_\e^{(t)}-z\bigr)^{-1}$ associated with a
modified metric graph subject to suitable vertex conditions. The estimate between the resolvents
of the two families is uniform with respect to the values of the ``spectral parameter'' $z$ in any compact $K\subset \mathbb C$ outside a fixed neighbourhood of a set $S:$
\begin{equation}
z\in K:
{\rm dist}(z, S)\ge\rho>0,
\label{z_set}
\end{equation}
where $S$ is the union of the limit spectrum for the family $A_\e^{(t)}$, described by \eqref{CC} (\emph{cf.} \cite{CherednichenkoCooper}), and the spectrum of the Dirichlet boundary-value problem on the ``soft'' component $[l_1, l_1+l_2).$ Following the same approach, it is possible to extend the results (at the expense of a worse estimate for the error term) to the transitional regime when $z\varepsilon^\omega$, $\omega<2,$ tends to a positive constant as $\varepsilon\to0.$
As for the ``high-frequency'' regime of $\omega=2$ ({\it cf.} \cite{Birman_internal_gap}, \cite{CKP} for the ``moderate-contrast'' high-frequency case),
the rationale of Section \ref{intermediate_comparison} is still applicable and leads to a different form of the result, which is outside the scope of the present paper.

In Section \ref{main_section} we carry out the uniform asymptotic analysis for the ``intermediate'' generalised resolvents of $\tilde A_\e^{(t)}$ in the
``finite-frequency'' setting, when the value of $z$ is fixed according to (\ref{z_set}) and $\varepsilon\to0.$ We prove our main result (Corollary \ref{main_result_statement}): for a suitable family $A_{\rm hom}^{(\tau)},$ $\tau:=\varepsilon t\in[0,2\pi),$ the bound
\begin{equation}
\bigl\Vert\bigl(A_\e^{(t)}-z\bigr)^{-1}-P\Phi_\varepsilon^*\bigl(\Psi^{(t)}\bigr)^*\bigl(A_{\rm hom}^{(\tau)}-z\bigr)^{-1}\Psi^{(t)}\Phi_\e P\bigr\Vert\le C\varepsilon^2, \ \ \ \ C>0,
\label{main_est}
\end{equation}
holds for all $\e\in(0,1],$ $t\in[0,2\pi\varepsilon^{-1}),$ and $z$ satisfying (\ref{z_set}), which yields, in particular, the spectral convergence result of \cite{CherednichenkoCooper}. We remark that in contrast to the result of \cite{Zhikov1989, BirmanSuslina} (``classical'' homogenisation) and \cite{CC} (multi-dimensional double porosity), where the error is estimated as $O(\e)$, in the case studied in the present paper it admits a higher order estimate.  In the estimate (\ref{main_est}), the projection operator $P$,
the unitary operators $\Phi_\varepsilon,$ $\Psi^{(t)},$
and the ``homogenised'' operator family $A_{\rm hom}^{(\tau)}$ are given explicitly, see formula (\ref{Phi_definition}) and Definitions \ref{Psi_definition}, \ref{Ahom}.

Finally, in Section 6 we show that the asymptotic behaviour given by the family $A_{\rm hom}^{(\tau)}$
is equivalently represented by a Schr\" odinger operator on ${\mathbb R}$ perturbed by a periodic dipole-type (``$\delta'$-type") potential. This suggests an interpretation of (\ref{orig_problem}) as a model of a ``metamaterial'', where the high contrast between components in the composite has an effective Kronig-Penney formulation with artificial magnetism.
The Kronig-Penney type effective description also suggests a strong connection between the problem (\ref{orig_problem})--(\ref{high_contrast_case})
and ``photonic band-gap materials": as the argument of Section \ref{Sect:Kronig_Penney} shows, the asymptotic result of the well-known work \cite{FK}, on $z$-dependent $\delta$-type interactions in periodic photonic crystals (albeit in a reduced Maxwell setting), is equivalent to the presence of
a $\delta'$-type interaction potential of the kind we obtain.

In what follows, we use interchangeably the notation $z$ and $k^2$ for the spectral parameter, as well as $\sqrt{z}$ and $k$ for the square root of it, where we always choose the branch so that
$\arg\sqrt{z}\in[0,\pi).$ For operators $A, B$ in a Hilbert space ${\mathfrak H},$ whenever we say that $Au=Bu+O(\varepsilon^2),$ $u\in{\mathfrak H},$ in the operator-norm sense as $\varepsilon\to0,$ we imply the existence of $C>0$ such that $\Vert Au-Bu\Vert\le C\varepsilon^2\Vert u\Vert$ for all $u\in {\mathfrak H}$ and $\varepsilon$ in some neighbourhood of zero.



\section{Gelfand transform, boundary triple, and  $M$-matrix}
\label{preliminaries1}


\subsection{Gelfand transform}
\label{Gelfand}

Consider
a graph ${\mathcal G}$ in ${\mathbb R}^d,$ invariant with respect to translations through elements of
${\mathbb Z}^d.$
For the one-dimensional Hausdorff measure $d{\mathcal H}^1$ on ${\mathcal G},$ we consider the space
$L^2({\mathcal G})$ of functions on ${\mathcal G}$ that are square integrable with respect to $d{\mathcal H}^1.$
We use the notation $Q:={\mathcal G}\cap[0,1)^d$  and $Q':=[0,2\pi)^d.$
The Gelfand transform, see \cite{Gelfand},
of a function $U\in L^2({\mathcal G})$ is the element  $\hat{U}=\hat{U}(y,\varkappa)$ of $L^2(Q\times Q')$ defined by the formula
\begin{equation}
\hat{U}(y,\varkappa)=(2\pi)^{-d/2}\sum_{n\in{\mathbb Z}^d}U(y+n)\exp\bigl(-{\rm i}\varkappa\cdot(y+n)\big),\ \ \ y\in
Q,
\ \varkappa\in Q'.
\label{Gelfand_formula}
\end{equation}
The Gelfand transform is a unitary operator between $L^2({\mathcal G})$ and $L^2(
Q\times Q'),$ where the inverse
transform is expressed by the formula
\begin{equation*}
U(y)=(2\pi)^{-d/2}\int_{Q'}\hat{U}(y, \varkappa)\exp({\rm i}\varkappa\cdot y)d\varkappa,\ \ \ \ y\in Q.
\end{equation*}
For the scaled version of the above transform, for $u\in L^2(\varepsilon{\mathcal G})$ we set
\begin{equation}
\hat{u}(x,t)=\biggl(\frac{\varepsilon}{2\pi}\biggr)^{d/2}\sum_{n\in{\mathbb Z}^d}u(x+\varepsilon n)\exp\bigl(-{\rm i}t\cdot(x+\varepsilon n)\big),\ \ \ x\in \varepsilon Q,\ t\in\varepsilon^{-1}Q',
\label{scaled_Gelfand}
\end{equation}
which is the result of applying the transform (\ref{Gelfand_formula}) to the function $U(y)=\varepsilon^{d/2}u(\varepsilon y)$ and setting
$y=x/\varepsilon.$ The inverse of the transform (\ref{scaled_Gelfand}) is given by
\begin{equation}
u(x)=\biggl(\frac{\varepsilon}{2\pi}\biggr)^{d/2}\int_{\varepsilon^{-1}Q'}\hat{u}(x,t)\exp({\rm i}t\cdot x)dt,\ \ \ \ x\in\varepsilon{\mathcal G}.
\label{inverse_scaled_Gelfand}
\end{equation}
In the rest of this article we use the above definitions with $d=1$ and consider the case of a connected graph ${\mathcal G},$ so that $Q=[0,1).$

Applying the above transform to the equation (\ref{orig_problem}) yields the direct fibre decomposition
\begin{equation}
-\frac{d}{dx}\biggl(a^\varepsilon(\cdot/\varepsilon)\frac{d}{dx}\biggr)=\int_{\oplus}\biggl(\dfrac{1}{\rm i}\dfrac{d}{dx}+t\biggr)a^\varepsilon\biggl(\dfrac{1}{\rm i}\dfrac{d}{dx}+t\biggr)\,dt,
\label{fibre_decomposition}
\end{equation}
where $\int_{\oplus}$ denotes the direct integral with respect to $t\in\bigl[0,2\pi\varepsilon^{-1}\bigr),$ and
all operators are defined in a standard way, {\it e.g.} by the corresponding sesquilinear forms.


\subsection{Boundary triples}
\label{triples}

Our approach is
based
on the theory of boundary triples \cite{Gor,Ko1,Koch,DM}, applied
to the class of operators introduced above. We next recall two fundamental
concepts of this theory, namely the boundary triple and the generalised Weyl-Titchmarsh
matrix function. Assume that $A_{\min}$ is a symmetric
densely defined operator with equal deficiency indices in a Hilbert space $H,$ and set $A_{\max}:=A_{\min}^*$.

\begin{definition}[\cite{Gor,Ko1,DM}]\label{Def_BoundTrip}
\label{definition1_1}
Let $\Gamma_0,$ $\Gamma_1$ be linear mappings of $\dom(A_{\max})$
to a separable Hilbert space $\mathcal{H}.$ The triple
$(\mathcal{H}, \Gamma_0,\Gamma_1)$ is called \emph{a boundary
triple} for the operator $A_{\max}$ if:
\begin{enumerate}
\item For all $u,v\in \dom(A_{\max})$ one has
\begin{equation}
\langle A_{\max} u,v \rangle_H -\langle u, A_{\max} v \rangle_H = \langle \Gamma_1 u, \Gamma_0
v\rangle_{\mathcal{H}}-\langle\Gamma_0 u, \Gamma_1 v\rangle_{\mathcal{H}}.
\label{Green_identity}
\end{equation}
\item The mapping
$u\longmapsto (\Gamma_0 u;
\Gamma_1 u),$ $f\in \dom(A_{\max})$ is surjective, {\it i.e.}, for all
$Y_0,Y_1\in\mathcal{H}$ there exists an element $y\in
\dom(A_{\max})$ such that $\Gamma_0 y=Y_0,\ \Gamma_1 y =Y_1.$
\end{enumerate}

A non-trivial extension ${A}_B$ of the operator $A_{\min}$ such
that $A_{\min}\subset  A_B\subset A_{\max}$  is called
\emph{almost solvable} if there exists a boundary triple
$(\mathcal{H}, \Gamma_0,\Gamma_1)$ for $A_{\max}$ and a bounded
linear operator $B$ defined on $\mathcal{H}$ such that for every
$u\in \dom(A_{\max})$
$$
u\in \dom({A_B})\ \ \ \text{\ if and only if }\ \ \ \Gamma_1 u=B\Gamma_0 u.
$$

The operator-valued function $M=M(z),$ defined by
\begin{equation*}\label{Eq_Func_Weyl}
M(z)\Gamma_0 u_{z}=\Gamma_1 u_{z}, \ \
u_{z}\in \ker (A_{\max}-z),\  \ z\in
\mathbb{C}_+\cup{\mathbb C}_-,
\end{equation*}
is called the Weyl-Titchmarsh $M$-function of the operator
$A_{\max}$ with respect to the corresponding boundary triple.
\end{definition}

The property of the  $M$-function that makes it the
tool of choice for the analysis of high-contrast periodic problems  is formulated as follows (\cite{DM,Ryzhov}):
provided that $A_B$ is an almost solvable extension of a
simple\footnote{In other words, there exists no reducing subspace $H_0$ such
that the restriction $A_{\min}|_{H_0}$ is a selfadjoint operator in
$H_0.$} symmetric operator $A_{\min}$, one has $z_0\in \rho(A_B)$ if
and only if $\bigl(B-M(z)\bigr)^{-1}$ admits analytic continuation
into 
$z_0$. Henceforth, we shall refer to points where the latter condition fails as ``zeros'' of
$B-M(z)$.

\subsection{The triple and the Green formula}
\label{triple_section}

For all $\varepsilon>0$ and $t\in[0,2\pi\varepsilon^{-1}),$ we study the operators $A_{\e}^{(t)}$ obtained by applying Gelfand transform
to the operator (\ref{orig_problem}), see (\ref{fibre_decomposition}). These are defined by the differential expressions
\begin{equation}
a_j\biggl(\frac 1{\rm i}\frac d{dx} +t\biggr)^2,\ \ \ \ j=1,2,3,\ \ \ \ \ \ a_2=\varepsilon^2,
\label{diff_expression}
\end{equation}
on the orthogonal sum $H_\varepsilon:=L^2(0,\ell_1)\oplus L^2(0,\ell_2)\oplus L^2(0,\ell_3),$ where
$\ell_j:=\e l_j,$ $j=1, 2, 3,$ so that
$\ell_1+\ell_2+\ell_3=\e.$ Here $l_1$ and $l_2$ are the same as in \eqref{high_contrast_case}, whereas $l_3:=1-l_1-l_2$.
The domain of the operator is the linear set in $\oplus_{j=1}^3 W^{2,2}(0,\ell_j)$ 
 of vector functions  $u=(u_1,u_2,u_3)^\top$ such that
\[
u_1(\ell_1)=u_2(0),\ \ \  u_2(\ell_2)=u_3(0),\ \ \ u_3(\ell_3)=u_1(0),
\]
\[
\dt u_1\vert_{\ell_1}= \dt u_2\vert_0,\ \ \  \dt u_2\vert_{\ell_2}=\dt u_3\vert_0,\ \ \  \dt u_3\vert_{\ell_3}=\dt u_1\vert_0.
\]
Here
\begin{equation}
\dt u:=a\biggl(\frac{du}{dx}+{\rm i}t u\biggr),
\label{partialt}
\end{equation}
where $a$ stands for $a_1$, $a_3,$ or $\e^2,$ depending on the interval that the derivative is taken on, see
(\ref{high_contrast_case}). Further, we define a normal derivative at the endpoints of each interval $[0, \ell _j],$ $j=1,2,3,$ in the direction towards the interior of the interval: 
\begin{equation}
\dnt u(x):= \begin{cases}
             \dt u(x), & \text{if } x=0,\\
             -\dt u(x), & \text{if } x=\ell_j.
            \end{cases}
\label{partialtn}
\end{equation}
The described operator can be viewed as defined
by the form
\[
\int_0^\varepsilon \bigl|\partial^{(t)}u\bigl|^2 dx
\]
considered on its natural domain.

By virtue of the fact that $A^{(t)}_\varepsilon$ is a family of problems on an interval viewed as a ``cycle'', where the end-points are identified with each other, it
proves convenient to exploit the toolbox of the theory of differential operators on metric graphs (``quantum graphs''), which we introduce next. In particular, in our treatment of the family
$A^{(t)}_\varepsilon,$ we build on a recent development of
the related theory in \cite{Yorzh3}, see also references therein, concerning the use of the $M$-function machinery in the study of the inverse spectral problem for quantum graphs.
Albeit not a familiar tool in homogenisation, the terminology and rationale of the theory of quantum graphs proves highly useful in addressing the behaviour of the related operator families.

With the above idea in mind, we view
$A_{\e}^{(t)}$ as a second-order differential operator on a
 metric graph ${\mathbb G}_\varepsilon$, which in our case is a simple cycle with three vertices, and rewrite the matching conditions in the following way. First, we identify the left endpoint of the interval $[0,\ell_j]$ with the right endpoint of the interval $[0,\ell_{j-1}],$ where for $j=1$ we set $j-1=3$.
 This yields three equivalence
classes of the edge endpoints, which we denote by $V_j,$ $j=1,2,3,$ while the matching conditions take the form:
$$
\forall\,j\ \ \  \ \ u \text{\ is continuous at\ } V_j,\ \
\sum_{x\in V_j} \dnt u(x)=0.
$$
We thus arrive at a ``quantum graph''
with an associated weighted magnetic Laplacian\footnote{The definition and well-known basic properties of the Laplacian on a quantum graph perturbed by a magnetic field are discussed
in {\it e.g.} \cite{Aharonov}, see also references therein.}, where all vertices are of the ``$\delta$-type'', using the terminology of \cite{Exner,Yorzh1}, with zero coupling constant at each vertex. In order to facilitate notation, we shall sometimes also denote by $x_m,$
$m=1,2, ..., 6,$ the endpoints of the intervals (graph edges) $\tilde e_j:=[0, \ell _j],$ $j=1,2,3,$ where the odd indices $m=1,3,5,$
correspond to the left end-points of the corresponding intervals, and the even indices $m=2,4,6,$ correspond to
their right end-points, respectively.

In
the spectral analysis of the above operator we use the boundary triples approach extensively.
First, we define a ``maximal'' operator
({\it cf.} \cite{Ryzhov}) $A_{\max}$ in the space $H_\varepsilon,$ by the same differential expression as above, its domain being $\oplus_j W^{2,2}(0,\ell_j),$ subject to the condition of continuity
at all vertices. We remark that the choice of the operator $A_{\rm max}$ is certainly non-unique, and for our choice
one has $A_{\text{max}}\subset A_{\text{max},0},$ where  $A_{\text{max},0}$ is defined on $\oplus_j W^{2,2}(0,\ell_j)$ and is adjoint to $A_{\text{min},0}$ defined on $W^{2,2}$-functions that vanish
with their first derivatives at the endpoints of all intervals $\tilde e_j,$ $j=1,2,3.$ Yet our choice turns out to be suitable for our purposes, as it leads to an ``effective'' boundary triple, using the terminology of \cite{Yorzh3}. We set the adjoint to $A_{\max}$ to be the ``minimal'' densely defined symmetric operator $A_{\min}.$  We choose the boundary triple
as follows: the boundary space is $\mathcal{H}=\mathbb{C}^3$, and the boundary operators are
\begin{equation*}
     (\Gamma_0u)_j:=  u(V_j),\ \ \ \ \
     (\Gamma_1u)_j:=
        \sum_{x\in V_j} \dnt u(x), \ \ \ \ \ j=1,2,3.
\end{equation*}
The Green identity (\ref{Green_identity})
holds by integration by parts:
\begin{multline*}
\langle A_{\rm max} u, v\rangle - \langle u, A_{\rm max} v\rangle
=\sum_{j=1}^{3} \left[ -u(x_{2j}) \dt \bar{v} (x_{2j}) + u(x_{2j-1})
\dt \bar{v} (x_{2j-1})  \right.\\[-0.5em]
\left. + \dt u(x_{2j})
\bar{v} (x_{2j})-\dt u(x_{2j-1}) \bar{v} (x_{2j-1}) \right ]= \sum_{k=1}^{6}
\left[ \dnt u(x_{k}) \bar{v} (x_{k}) - u(x_{k}) \dnt
\bar{v} (x_{k})
 \right ].
\end{multline*}
Rearranging the sum in the last expression yields
\begin{multline*}
\langle A_{\rm max} u, v\rangle - \langle u, A_{\rm max}
v\rangle
=\sum_{j=1}^3 \biggl(\sum_{k:\,x_k\in V_j}
\dnt u(x_{k}) \bar{v} (x_{j}) - \sum_{k:\,x_k\in V_j}
u(x_{j})
\dnt \bar{v}
(x_{k})\biggr)
 =\langle \Gamma_1 u, \Gamma_0 v\rangle_{\mathbb{C}^3} - \langle \Gamma_0 u, \Gamma_1 v \rangle_{\mathbb{C}^3},
\end{multline*}
as required.

\subsection{Datta -- Das Sarma conditions}
\label{Datta_section}

In what follows, we study second-order differential operators on metric graphs with matching conditions more general than those of $\delta-$type, introduced above, namely, with the so-called weighted, or ``Datta -- Das Sarma'', matching conditions, see \cite{pavlov1,pavlov2,Datta}. In the case of differential expression (\ref{diff_expression}) on the graph ${\mathbb G}_\varepsilon,$ the corresponding modification is described as follows.

Assume that some endpoints $x_m$ are assigned weights $w_m$ such that either $w_m\in{\mathbb R}$ or
$w_m=\exp({\rm i}\theta_m),$ $\theta_m\in{\mathbb R}.$ Without loss of generality, we set $w_m=1$ for all remaining endpoints $x_m.$ Then the formulae at the end of Section \ref{triple_section} stand, if one modifies the definition of the domain of $A_{\max}$ and the definition of boundary operators $\Gamma_0^{(1)}, \Gamma_1^{(1)},$ as follows. The domain of the new operator $A_{\max}$ consists of all $W^{2,2}$-functions $u$ such that $w_l u(x_l)=w_k u(x_k)$ for all $x_k, x_l\in V_j,$
and
$$
\bigl(\Gamma_0^{(1)}u\bigr)_j:=
w_k u(x_k), \ \
x_k\in V_j,\ \ \ \ \ \
\bigl(\Gamma_1^{(1)}u\bigr)_j:= \sum_{x_k\in V_j} \widehat{\partial}_n^{(t)}u(x_k),\ \ \ j=1,2,3,
$$
where
\[
\widehat{\partial}_n^{(t)}u(x_k):=\overline{w_k}\,^{-1}\dnt u(x_k),\ \ \ \ \ \ \ \ \ k=1, 2,..., 6.
\]
Introducing the weights described above allows for the treatment of graph operators with more general matching conditions than the basic
$\delta$-type conditions. In particular, the analysis is no longer limited
to domains consisting of functions that are either continuous or have continuous co-normal derivatives.

In what follows, it is crucial that we can consider matching conditions that no longer have zero coupling constants, or equivalently in terms of the boundary operators introduced above, that are no longer described as $\Gamma_1 u=0$ on the domain of $A_{\max}$. We parameterise these general matching conditions by a matrix $B,$ {\it cf.} Definition \ref{definition1_1}. For each operator and boundary triple considered, we attach a superscript to the related matrices $B$ and $M,$ so that the matrices with the same superscript always pertain to the same operator and the same triple.


\subsection{$M$-matrix}
\label{Mmatrixsection}
In order to proceed with the spectral analysis of the operator family $A_{\e}^{(t)}$ introduced above, we construct its  $M$-matrix with respect to  the boundary triple described in Section \ref{triple_section}. On all edges of the graph we deal with a differential equation of the form
\begin{equation}
\label{geneq}
-a\biggl(\frac{d}{dx}+{\rm i}t\biggr)^2u
=k^2u,
\end{equation}
with a suitable value of the coefficient $a=a_j>0,$ $j=1,2,3,$ where $a_2=\varepsilon^2.$
For any solution $u$ of the equation (\ref{geneq}) on the interval $[0,l]$ one has
\begin{equation}
u=A{\rm e}^{-{\rm i}tx}\exp\biggl(-{\rm i}\frac{k}{\sqrt{a}}x\biggr)+B{\rm e}^{-{\rm i}tx}\exp\biggl({\rm i}\frac{k}{\sqrt{a}}x\biggr)
\label{genform}
\end{equation}
with some $A, B\in{\mathbb C}.$
The solution $u$ such that
$u(0)=1,$ $u(l)=0,$ corresponds to the values
\[
A=\biggl\{2{\rm i}\sin\biggl(\dfrac{k}{\sqrt{a}}l\biggr)\biggr\}^{-1}\exp\biggl({\rm i}\dfrac{k}{\sqrt{a}}l\biggr),\ \ \ \ \ \ \ \ \ \
B=-\biggl\{2{\rm i}\sin\biggl(\dfrac{k}{\sqrt{a}}l\biggr)\biggr\}^{-1}\exp\biggl(-{\rm i}\dfrac{k}{\sqrt{a}}l\biggr).
\]

Consider a vertex of ${\mathbb G}_\varepsilon,$ such that one of its adjacent edges is represented by the above interval $[0,l]$ that ``starts'' at the vertex,
{\it i.e.} the vertex is represented by the boundary point $0.$
Then the contribution at the vertex to the value of the boundary operator $\Gamma_1^{(1)}$ calculated on the solution
(\ref{genform})
is
given by
\begin{equation}
\partial_n^{(t)}
u(x)\Bigr\vert_{x=0}=A\biggl\{a\biggl(-{\rm i}t-{\rm i}\frac{k}{\sqrt{a}}\biggr)+{\rm i}ta\biggr\}{\rm e}^{-{\rm i}tx}\exp\biggl(-{\rm i}\frac{k}{\sqrt{a}}x\biggr)
\Biggr\vert_{x=0}
\end{equation}
\begin{equation}
+B\biggl\{a\biggl(-{\rm i}t+{\rm i}\frac{k}{\sqrt{a}}\biggr)+{\rm i}ta\biggr\}{\rm e}^{-{\rm i}tx}\exp\biggl({\rm i}\frac{k}{\sqrt{a}}x\biggr)\Biggr\vert_{x=0}
\label{firstAB}
\end{equation}
\[
=-A{\rm i}k\sqrt{a}{\rm e}^{-{\rm i}tx}\exp\biggl(-{\rm i}\frac{k}{\sqrt{a}}x\biggr)\Biggr\vert_{x=0}+B{\rm i}k\sqrt{a}{\rm e}^{-{\rm i}tx}\exp\biggl({\rm i}\frac{k}{\sqrt{a}}x\biggr)\Biggr\vert_{x=0}
=-k\sqrt{a}\cot\biggl(\dfrac{k}{\sqrt{a}}l\biggr).
\]
A similar contribution of the boundary operator $\Gamma_1^{(1)}$ for the case of an edge that ``terminates'' at the vertex, {\it i.e.} the vertex is represented
by  the boundary point $l$ is given by
\[
\partial_n^{(t)}u(x)\Bigr\vert_{x=l}=A\biggl\{a\biggl({\rm i}t\frac{k}{\sqrt{a}}\biggr)-{\rm i}ta\biggr\}{\rm e}^{-{\rm i}tx}\exp\biggl(-{\rm i}\frac{k}{\sqrt{a}}x\biggr)\biggr\vert_{x=l}
\]
\[
+B\biggl\{a\biggl({\rm i}t-{\rm i}\frac{k}{\sqrt{a}}\biggr)-{\rm i}ta\biggr\}{\rm e}^{-{\rm i}tx}\exp\biggl({\rm i}\frac{k}{\sqrt{a}}x\biggr)\Biggr\vert_{x=l}
\]
\[
=A{\rm i}k\sqrt{a}{\rm e}^{-{\rm i}tx}\exp\biggl(-{\rm i}\frac{k}{\sqrt{a}}x\biggr)\Biggr\vert_{x=l}+B{\rm i}k\sqrt{a}{\rm e}^{-{\rm i}tx}\exp\biggl({\rm i}\frac{k}{\sqrt{a}}x\biggr)\Biggr\vert_{x=l}
=k\sqrt{a}{\rm e}^{-{\rm i}tl}\csc\biggl(\dfrac{k}{\sqrt{a}}l\biggr).
\]
Therefore, the following explicit formula for the $M$-matrix holds:
\begin{equation}
M^{(1)}_\varepsilon(z)
=\left(
\begin{array}{ccc}
 -\sum\limits_{j=1,3}\sqrt{a_j} k \cot \dfrac{k \ell_j}{\sqrt{a_j}} & \sqrt{a_1} {\rm e}^{{\rm i} \ell_1 t} k \csc \dfrac{k \ell_1}{\sqrt{a_1}} & \sqrt{a_3} {\rm e}^{-{\rm i} \ell_3 t} k \csc\dfrac{k \ell_3}{\sqrt{a_3}}
 \\[1em]
 \sqrt{a_1} {\rm e}^{-{\rm i} \ell_1 t} k \csc \dfrac{k \ell_1}{\sqrt{a_1}} & -\sum\limits_{j=1,2}\sqrt{a_j} k \cot \dfrac{k \ell_j}{\sqrt{a_j}} & \sqrt{a_2} {\rm e}^{{\rm i} \ell_2 t} k \csc \dfrac{k \ell_2}{\sqrt{a_2}}
 \\[1em]
 \sqrt{a_3} {\rm e}^{{\rm i} \ell_3 t} k \csc \dfrac{k \ell_3}{\sqrt{a_3}} & \sqrt{a_2} {\rm e}^{-{\rm i} \ell_2 t} k \csc \dfrac{k \ell_2}{\sqrt{a_2}} & -\sum\limits_{j=2,3}\sqrt{a_j} k \cot \dfrac{k \ell_j}{\sqrt{a_j}}
\end{array}
\right).
\label{M^(1)}
\end{equation}

\subsection{Zeros of the $M$-matrix and spectrum}\label{sect:motivation}


Putting the discussion about simplicity of $A_{\min}$ aside for a moment, consider the set of ``zeroes'' of
 $M^{(1)}_\varepsilon,$ which we as above define as those points $z$ at which
 $M^{(1)}_\varepsilon(z)$ has a zero eigenvalue.

\begin{proposition}
\label{detM_proposition}
The determinant of $M^{(1)}_\varepsilon(z)$ admits the following asymptotic formula as $\e\to0$ for all $z\equiv k^2\in K,$ where $K\subset\mathbb C$ is a compact:
\begin{equation}
\det M^{(1)}_\varepsilon
(z)=(l_1 l_3\e)^{-1}a_1 a_3 k \bigl(2\csc k l_2 \cos \e t+k (l_1+l_3)-2 \cot k l_2\bigr) + O(\e).
\label{detM_expansion}
\end{equation}
\end{proposition}

\begin{proof}

We substitute $\ell_j=\e l_j,$ $a_2=\e^2$ into (\ref{M^(1)}) and expand trigonometric functions into power series
wherever possible. Note, that since $t$ is not bounded independently of $\e$ (indeed, $t$ spans the interval
$\bigl[0,2\pi\e^{-1}\bigr)$, which grows as $\e\to 0$), one cannot use power expansions for exponentials. As a result, we obtain the following formula:
$$
M^{(1)}_\varepsilon(z)=
\left(
\begin{array}{cc}
-\dfrac{1}{\varepsilon}\biggl(\dfrac{a_1}{l_1}+\dfrac{a_3}{l_3}\biggr)+\dfrac{1}{3} \e  \left(l_1 k^2+l_3 k^2\right)
& \dfrac{a_1}{\e  l_1}{\rm e}^{{\rm i} \e t l_1}+\dfrac{1}{6}{\rm e}^{{\rm i} \e t l_1} \e  l_1 k^2  \\[1.0em]
 \dfrac{a_1}{\e  l_1}{\rm e}^{-{\rm i} \e t l_1} +\dfrac{1}{6}{\rm e}^{-{\rm i} \e t l_1} \e  l_1 k^2&
-\dfrac{a_1}{\e  l_1}+\e\left(\dfrac{k^2 l_1}{3} -k \cot \left(k l_2\right)\right) \\[1.0em]
 \dfrac{a_3}{\e  l_3}{\rm e}^{{\rm i} \e t l_3}+\dfrac{1}{6}{\rm e}^{{\rm i} \e t l_3} \e  l_3 k^2 &
 {\rm e}^{-{\rm i} \e t l_2} k \e  \csc \left(k l_2\right)
\end{array}\right.\ \ \ \ \ \ \ \ \ \ \ \ \ \ \ \ \ \ \ \ \ \ \ \ \ \ \ \ \ \ \ \ \ \ \ \ \ \ \ \
$$
\[
\ \ \ \ \ \ \ \ \ \ \ \ \ \ \ \ \ \ \ \ \ \ \ \ \ \ \ \ \ \ \ \ \ \ \ \ \ \ \ \ \ \ \ \ \ \ \ \ \ \ \ \ \ \ \ \ \ \ \ \ \ \ \left.
\begin{array}{c}
\dfrac{a_3}{\e  l_3}{\rm e}^{-{\rm i} \e t l_3}+\dfrac{1}{6}{\rm e}^{-{\rm i} \e t l_3} \e  l_3 k^2\\[1.0em]
{\rm e}^{{\rm i} \e t l_2} k \e  \csc \left(k l_2\right)\\[1.0em]
-\dfrac{a_3}{\e  l_3}+\e\left(\dfrac{k^2 l_3}{3}-k \cot \left(k l_2\right)\right)
\end{array}
\right)+O\bigl(\varepsilon^3\bigr),
\]
as $\varepsilon\to0,$ and (\ref{detM_expansion}) follows.
\end{proof}

The spectrum of the operator $A_{\e}^{(t)}$ is a union of the set $S_M^\e$ of points $z$ into which the inverse of $M^{(1)}_\varepsilon$ can not be analytically continued (zeroes of $M^{(1)}_\varepsilon$) and the set $S_{\min}$ of eigenvalues
of the reducing self-adjoint ``part''
of the symmetric minimal operator $A_{\text{min}}=A_{\max}^*,$ which are ``invisible'' to the $M$-matrix, as discussed
in {\it e.g.} \cite{DM}. The latter appear whenever the operator $A_{\text{min}}$ is not simple, {\it cf.} Section \ref{triples} above. A straightforward argument, see {\it e.g.} \cite{Yorzh1}, demonstrates that in our case $S_{\min}$ coincides with the set of eigenvalues of the symmetric operator  $A_{\min}$. In our setting, the named operator is defined by the same differential expression as $A_{\text{max}}$ on functions $u\in\dom(A_{\max})$ subject to the conditions $\Gamma_0 u=\Gamma_1 u=0$.

Proposition \ref{detM_proposition} immediately implies
that for all compact $K\subset \mathbb C,$ the set $S_M^\e\cap K$
 converges as $\e\to0$ to the set of solutions $k^2\in K$ to
\begin{equation}\label{CC}
2\cos\tau+k (l_1+l_3)\sin k l_2-2 \cos k l_2=0,\ \ \ \tau=\varepsilon t\in[0,2\pi),
\end{equation}
in line with the result of \cite{CherednichenkoCooper}.  Notice that for each $\varepsilon, t,$ the set of poles of $M^{(1)}_{\e}$, where one needs to check additionally
whether $M^{(1)}_{\e}$ has a vanishing eigenvalue, coincides with the set of zeroes of $\sin kl_2$, at which the determinant (\ref{detM_expansion}) is either regular or has a pole. It is regular at a given point in this set if and only if one has
$\vert\cos\e t\,\vert=1$ at the same time ({\it i.e.} $t=0$ or $t=\pi/\varepsilon$), which immediately implies that
exactly one eigenvalue of
$M^{(1)}_\e$ vanishes for such $k,$ $\varepsilon,$ $t.$ Clearly, these values of $k,$ $\varepsilon,$ $t$ also satisfy (\ref{CC}).

In the remainder of this section, motivated by the above calculation, we give an example of an operator family that is asymptotically isospectral 
(as $\varepsilon\to0$) but is not resolvent-close to the family $A_\varepsilon^{(t)}.$ For all $z\in{\mathbb R}_+,$ define the operator family $\check A^{(\check{\tau})}(z)$ by the differential expression
\[
\biggl(\frac 1{\rm i} \frac d{dx} +\check{\tau}\biggr)^2,\qquad \check{\tau}\in\bigl[0,2\pi l_2^{-1}\bigr),
\]
on the interval $[0,l_2]$ with the following $z$-dependent
conditions:
\begin{equation}\label{eq:z_dep}
u(0)=u(l_2),\ \ \ \ \ \partial^{(\check{\tau})}_n u\bigr\vert_0+\partial^{(\check{\tau})}_n u\bigr\vert_{l_2}=-z(l_1+l_3)u(0).
\end{equation}
Here, notation analogous to (\ref{partialt})--(\ref{partialtn}) is used:
\begin{equation}\label{eq:tau-notation}
\partial^{(\tau)}u:=\frac{du}{dx}+{\rm i}\tau u,\ \ \ \  \ \ \ \ \  \ \ \ \
\partial^{(\tau)}_n u(x):= \begin{cases}
             \partial^{(\tau)}u(x), & \text{if } x=0,\\
             -\partial^{(\tau)}u(x), & \text{if } x=l_2,
            \end{cases}\ \ \ \ \ \ \ \ \ \ \ \tau\in{\mathbb R},
\end{equation}
with $\tau=\check{\tau}.$ We remark that
  $\check A^{(\check{\tau})}(z)$ can be treated as an operator pencil, admitting the form of a differential operator with an energy-dependent perturbation that is a Dirac delta-function multiplied by a spectral parameter, see \cite{KuchmentZeng2004, Exner} and Section \ref{discussion}. It is checked directly that the set of $z=k^2$ such that $k$ is a solution to \eqref{CC} coincides with the set of poles of the resolvent\footnote{Note that we evaluate the resolvent at the point $k^2,$ which determines the domain of the operator. The object thus defined is therefore a generalised resolvent, {\it cf.} Section \ref{sect:Krein} below.} $\bigl(\check A^{(\check{\tau})}(z)-z\bigr)^{-1}.$
Indeed, consider a cycle of two vertices connected by two edges of lengths $\check l_1,$ $\check l_2,$ such that $\check l_1+\check l_2=l_2$. Proceeding as above yields the following $M$-matrix for the operator $\check A_{\max}$ on the domain of $W^{2,2}$-functions that are
  continuous on the cycle:
$$
\check M^{(\check{\tau})}(z)=k\begin{pmatrix}
-\cot {k \check l_1}-\cot k \check l_2 & \dfrac {{\rm e}^{{\rm i}\check{\tau} \check l_1}}{\sin k\check l_1}+\dfrac {{\rm e}^{-{\rm i}\check{\tau} \check l_2}}{\sin k\check l_2}\\[1.0em]
\dfrac {{\rm e}^{-{\rm i}\check{\tau} \check l_1}}{\sin k\check l_1}+\dfrac {{\rm e}^{{\rm i}\check{\tau} \check l_2}}{\sin k\check l_2}
& -\cot {k \check l_1}-\cot k \check l_2
\end{pmatrix}.
$$
The requirement that at one of the vertices, say $V_1,$ one has the energy-dependent matching condition \eqref{eq:z_dep}, leads to the equation
$$
\det\bigl(\check M^{(\check{\tau})}(z)- \check B(z)\bigr)=0,\ \ \ \ \
\check B(z):=\mathrm{diag}\bigl\{-(l_1+l_3)z,0\bigr\},
$$
which by a straightforward
manipulation is reduced to \eqref{CC}, with $\check{\tau}=\tau/l_2.$

The above argument shows that (the ``visible'' part of) the spectra of the family $A_\e^{(\tau/\varepsilon)}$ converge, as $\varepsilon\to0,$ to the set of singularities of the generalised resolvent  $\bigl(\check A^{(\tau/l_2)}({z})-z\bigr)^{-1}$, which
suggests
that
$\check A^{(\tau/l_2)}({z})$
 is the resolvent limit of the family $A_\e^{(t)}$ in the operator-norm
 sense.
However, as we
demonstrate below (see Theorem \ref{second_claim} and Remark \ref{remark_label}),
this is false for $\bigl(\check A^{(\tau/l_2)}({z})-z\bigr)^{-1}$ or
any of its unitary transformations, and
a closely related self-adjoint operator, albeit in a larger space,
has the desired property.

\section{Preliminary observations}
\label{preliminary}

\subsection{Auxiliary re-scaling in the soft component}
\label{section_rescaling}
Motivated by the above result on spectral convergence, we apply to the initial operator family $A_\e^{(t)}$
a unitary transformation that
rescales the soft component interval $[0,\e l_2]$ to size of order one
while leaving the stiff component intact. The unitary image of $A_\e^{(t)}$ under this transformation is shown to have the same $M$-matrix as the operator $A_\e^{(t)},$ after an appropriate modification of the boundary triple. This modification
 is done by passing from the $\delta$-type coupling to a Datta -- Das Sarma coupling at the endpoints of the interval $[0,l_2]$. To this end, we consider the unitary dilation $F_\e: L^2(0,\e l_2)\to L^2(0,l_2),$ given by
$(F_\e u)(x):=\sqrt{\e} u(\e x)$. The operator
\begin{equation}
\Phi_\e:=\begin{pmatrix}
I &0&0\\
0&F_\e&0\\
0&0&I
\end{pmatrix}
\label{Phi_definition}
\end{equation}
is a unitary transform of $\oplus_j L^2(0,\e l_j)$ to the space $H:=L^2(0,\e l_1)\oplus L^2(0,l_2)\oplus L^2(0,\e l_3)$. We denote by $\mathbb{G}$ the graph $\mathbb{G}_\varepsilon$, to which the above rescaling has been applied. Clearly, the matching conditions at the vertex common to $[0,\e l_1]$ and $[0,\e l_3]$ are not affected. As for the matching conditions at the remaining vertices $V_2$ and $V_3$ (see Fig.\,\ref{fig:mod}(a) on p.\,\pageref{fig:mod} below), the following calculation applies. First, notice that the differential expression on the ``developed'' weak component
\begin{equation}\label{eq:soft}
F_\e\biggl(\frac{\e}{{\rm i}}\frac {d}{dx}+\e t\biggr)^2 F_\e^*u
=\biggl(\frac{1}{\rm i}\frac {d}{dx} +t\e\biggr)^2 u
\end{equation}
remains essentially the same, with the symbol of the differential part of the operator losing the coefficient $\e$.
As for the endpoints of the dilated soft component, they acquire Datta -- Das Sarma weights $1/\sqrt{\e}$. This is immediately obvious for the values of the function under the unitary transformation $F_\e$, whereas for $\dt u$ one has:
$$
\dt u = \e^2 \frac 1{\e^{3/2}}(F_\e u)'+{\rm i}\e^2 t \frac 1{\sqrt{\e}} (F_\e u)=\sqrt{\e}\bigl((F_\e u)'+{\rm i}t \e (F_\e u)\bigr)=
\sqrt{\e}\dtau (F_\e u),
$$
where $\tau=\e t$ and the notation \eqref{eq:tau-notation} is
used.

In line with the discussion of Section \ref{Datta_section}, the boundary triple for the rescaled operator is chosen as follows: both endpoints of the interval $[0,l_2]$ are  assigned the weight $w_m=1/\sqrt{\e}$, whereas
 $w_m=1$ for all remaining endpoints $x_m.$ The domain of $A_{\max}$ consists of all $W^{2,2}$-functions $u$ such that $w_l u(x_l)=w_k u(x_k)$ for all $x_l,x_k\in V_j$ for each vertex $V_j,$ $j=1, 2, 3,$ and
$$
\bigl(\Gamma_0^{(2)}u\bigr)_j:=
w_k u(x_k),\ \
x_k\in V_j,\ \ \ \ \
\bigl(\Gamma_1^{(2)}u\bigr)_j:= \sum_{x_k\in V_j} \widehat{\partial}_n^{\sigma(w_k)} u(x_k),\ \ \ j=1, 2, 3,
$$
where ({\it cf.} (\ref{partialt})--(\ref{partialtn}), (\ref{eq:tau-notation}))
$$\widehat{\partial}_n^{\sigma(w_k)} u(x_k):=\overline{w_k}\,^{-1}\partial_n^{\sigma(w_k)} u(x_k),
\ \ \ \ \ \ \ \ \ \ \ \sigma(w_k):=\begin{cases}
(t),& \text{ if } w_k=1,\\
(\tau), & \text { if } w_k=1/\sqrt{\e},
\end{cases}\ \ \ \ \ \ \ \ \ \ \ k=1,2, ...,6.
$$

\begin{remark}
The formula \eqref{eq:soft} suggests that after the unitary rescaling $\Phi_\e,$ the differential expression that defines the operator loses its dependence on the parameter $\e$ on the soft component. This becomes obvious after the substitution $\tau =\e t$ in \eqref{eq:soft}. Henceforth, we use $\tau$ and $\e t$ interchangeably:
 $\tau$ in the objects pertaining to the soft component, and $\e t$ in those pertaining to the stiff component, as in the latter case one cannot drop the explicit dependence on $\e$.
\end{remark}

The claim concerning the form of the $M$-matrix follows. Indeed, when obtaining its expression one constructs for any given vertex $V$ the solution $u_z\in{\rm ker}(A_{\max}-z)$ such that this solution equals unity at the vertex $V$ and  zero at any other vertex ({\it cf.} Section \ref{triple_section} above). Such solutions are constructed independently on any edge emanating from the vertex $V.$ If this edge is the edge $[0,l_2]$, the corresponding solution acquires the
factor $\sqrt{\e}$ compared to the corresponding solution on the edge $[0,\e l_2]$. The column of the $M$-matrix corresponding to the vertex $V$ is then obtained by evaluating either $\dnt u_z$ or $\sqrt{\e}\,\dntau u_z$, which yields yet another multiplication by $\sqrt{\e}$ of the normal derivatives at both endpoints of $[0,l_2]$, where we use the fact that Datta -- Das Sarma weights at the two endpoints are equal. As a result, we obtain the following expression for the $M$-matrix of the unitary image of the operator $A_\e^{(t)},$ {\it cf.} (\ref{M^(1)}):
\begin{equation}
M^{(2)}_\varepsilon(z)=\left(
\begin{array}{ccc}
 -\sum_{j=1,3}\sqrt{a_j} k \cot \dfrac{k \e l_j}{\sqrt{a_j}} & \sqrt{a_1} {\rm e}^{{\rm i} \e l_1 t} k \csc \dfrac{k  \e l_1}{\sqrt{a_1}} & \sqrt{a_3} {\rm e}^{-{\rm i}  l_3 \e t} k \csc\dfrac{k \e l_3}{\sqrt{a_3}} \\[1.1em]
 \sqrt{a_1} {\rm e}^{-{\rm i} \e l_1 t} k \csc \dfrac{k \e l_1}{\sqrt{a_1}} & -\sqrt{a_1} k \cot \dfrac{k \e l_1}{\sqrt{a_1}}-\e k \cot{k l_2}  & \e {\rm e}^{{\rm i} \e l_2  t} k \csc k  l_2  \\ [1.1em]
 \sqrt{a_3} {\rm e}^{{\rm i} \e l_3 t} k \csc \dfrac{k \e l_3}{\sqrt{a_3}} & \e  {\rm e}^{-{\rm i} \e l_2 t} k \csc k  l_2 & -\sqrt{a_3} k \cot \dfrac{k \e l_3}{\sqrt{a_3}}-\e k \cot{k l_2} \end{array}
\right).
\label{M2form}
\end{equation}

\subsection{Kre\u\i n resolvent formula}
\label{sect:Krein}
One of the cornerstones of our analysis is the celebrated Kre\u\i n formula, which allows to relate the resolvent of $A_B,$ see Section \ref{triples}, to the resolvent of a self-adjoint operator $A_\infty$ defined as the
restriction of the maximal operator $A_{\text{\rm max}}$ to the set
$$
\dom(A_\infty)=\bigl\{u\in \dom(A_{\text{\rm max}})|\, \Gamma_0 u=0\bigr\}.
$$
(We follow \cite{Ryzhov} in using the notation $A_\infty,$ justified by the fact that in the language of triples this extension formally corresponds to $A_B$ with $B=\infty.$)

In particular, we will find it necessary to consider not only proper operator extensions $A_B$ of the symmetric operator $A_{\text{min}}$ which are defined on domains
$$
\dom(A_B)=\bigl\{u\in \dom(A_{\text{\rm max}})|\, \Gamma_1 u=B\Gamma_0 u\bigr\}
$$
parameterised by bounded in $\mathcal{H}$ operators $B$, but also those for which the parameterising operator $B$ depends on the spectral parameter $z$. This amounts to considering spectral boundary-value problems where the spectral parameter is present not only in the differential equation but also in the boundary conditions:
\begin{equation}
A_{\text{\rm max}}u-z u =f,\ \ \ \
u\in \dom (A_{\text{\rm max}}),\ \ \
\Gamma_1 u = B(z) \Gamma_0 u.
\label{gen_prob}
\end{equation}
The solution operator $R(z)$ for a boundary-value problem of this type is known \cite{Strauss} to be a generalised resolvent in the case when $-B(z)$ is an $R$-function:
if $B(z)$ is analytic in $\mathbb C_+\cup \mathbb C_-$ with $\Im z \Im B(z)\leq 0,$ then
\begin{equation}\label{eq:out-of-space}
R(z)=P_{\mathfrak H} (A_{\mathfrak H}-z)^{-1}\bigr|_{\mathfrak H},
\end{equation}
where $\mathfrak H$ is a Hilbert space such that $H\subset \mathfrak H,$ the operator $P_{\mathfrak H}$ is the orthogonal projection of $\mathfrak H$ onto $H,$ and $A_{\mathfrak H}$ is a self-adjoint in $\mathfrak H$ out-of-space extension of the operator $A_{\text{min}}$.

On the other hand, for any fixed $z$ the operator $R(z)$  coincides with the resolvent (evaluated at the point $z$) of a closed linear operator 
that is a proper extension of $A_{\text{min}}$ with the $z$-dependent domain given in (\ref{gen_prob}).
It is for this reason that in what follows we preserve the notation $(A_B-z)^{-1}$ for the generalised resolvent of $A_B$ when 
$B=B(z).$

The Kre\u\i n formula suitable for treatment of such problems was obtained in \cite{DM}. For the sake of completeness we include a short proof of this result.

\begin{proposition}[Version of the Kre\u\i n formula of \cite{DM}]\label{prop:Krein}
Assume that $\{\mathcal{H},\Gamma_0,\Gamma_1\}$ is a boundary triple for the operator $A_{\text{\rm max}}$. Then for the (generalised) resolvent  $(A_B-z)^{-1}$, where $B=B(z)$ is a bounded operator in $\mathcal{H}$ for $z\in\mathbb C_+\cup \mathbb C_-$, one has, for all
$z\in \rho(A_B)\cap \rho(A_\infty)$:
\begin{multline}
\label{eq:resolvent}
(A_B-z)^{-1}=(A_\infty-z)^{-1}+ \gamma(z)\bigl(B(z)-M(z)\bigr)^{-1}\gamma^*(\bar z)
\\
=(A_\infty-z)^{-1}+ \gamma(z)\bigl(B(z)-M(z)\bigr)^{-1}\Gamma_1 (A_\infty-z)^{-1},
\end{multline}
where $M(z)$ is the  M-function of $A_{\text{\rm max}}$ with respect to the boundary triple $\{\mathcal{H},\Gamma_0,\Gamma_1\}$ and $\gamma(z)$ is the solution operator
$$
\gamma(z)=\bigl(\Gamma_0|_{\text{\rm ker\,}(A_{\text{\rm max}}-z)}\bigr)^{-1}.
$$
\end{proposition}

\begin{proof}
For any $f\in H$, one clearly has
\begin{equation}\label{eq:K1}
u_z:=(A_B-z)^{-1}f-(A_\infty-z)^{-1}f\in \text{ker\,}(A_{\text{\rm max}}-z):=N_z
\end{equation}
Setting $u:=(A_B-z)^{-1}f$ and using the explicit description of the domain of $A_\infty$ together with the equality \cite{DM} $\gamma^*(\bar z)=\Gamma_1 (A_\infty -z)^{-1}$, one has:
\begin{equation*}
\Gamma_1 u = \Gamma_1 u_z + \Gamma_1 (A_\infty -z)^{-1} f= \Gamma_1 u_z + \gamma^*(\bar z)f, \ \ \ \ \ \
\Gamma_0 u = \Gamma_0 u_z,
\end{equation*}
and, since
$$
\Gamma_1 u=B(z) \Gamma_0 u,
$$
one immediately arrives at the equality
$$
\Gamma_1 u_z + \gamma^*(\bar z) f = B(z) \Gamma_0 u_z.
$$
On the other hand, since $u_z\in N_z$ one has $\Gamma_1 u_z=M(z) \Gamma_0 u_z$, which yields
$$
\bigl(B(z)-M(z)\bigr)\Gamma_0 u_z = \gamma^*(\bar z)f,
$$
and hence
$$
\Gamma_0 u_z =\bigl(B(z)-M(z)\bigr)^{-1} \gamma^*(\bar z)f.
$$
Since $\Gamma_0$ is invertible \cite{DM} on $N_z$ provided that $z\in\rho(A_\infty),$ and writing $(\Gamma_0|_{N_z})^{-1}=\gamma(z)$, this leads to
$$
u_z=\gamma(z)\bigl(B(z)-M(z)\bigr)^{-1}\gamma^*(\bar z)f,
$$
which together with \eqref{eq:K1} completes the proof.
\end{proof}

\section{Comparison to the  ``intermediate'' generalised resolvents $\bigl(\tilde{A}_\varepsilon^{(t)}-z\bigr)^{-1}$}
\label{intermediate_comparison}

We shall now consider an operator family $\tilde A_\e^{(t)}$ that is defined by the same differential expression as
$\Phi_\varepsilon A_\e^{(t)}\Phi_\varepsilon^*$ and on the same Hilbert space $H$ but is different from
$\Phi_\varepsilon A_\e^{(t)}\Phi_\varepsilon^*$ as a graph Hamiltonian: it is defined by a topologically different underlying metric graph $\widetilde{\mathbb G}$ in the terminology of the spectral theory of quantum graphs. The graph $\widetilde{\mathbb G}$ has two components that correspond to the ``soft'' and ``stiff'' components of the original graph ${\mathbb G}.$ These will be almost decoupled, but for the non-local interface condition of the order $\sqrt{\e}$ intertwining the two. This family turns out to be a good approximation, up to a rank-one operator, for the original operator family $A_\e^{(t)}$, while being at the same time 
a convenient intermediate operator for the
final step of our plan, the passage to the homogenised operator.
\emph{From now on, we shall assume that $a_1=a_3\equiv a$ for the sake of brevity}.
Note that the domain of $\tilde A_\e^{(t)}$ depends on the spectral parameter $z.$ The operator $\bigl(\tilde A_\e^{(t)}-z\bigr)^{-1}$
solves a spectral boundary-value problem where the spectral parameter is present not only in the differential equation, but also in the associated boundary conditions. In the terminology of \cite{DM,Strauss}, it is therefore a generalised resolvent of the corresponding boundary-value problem, \emph{cf.} Section \ref{sect:Krein} above. Nevertheless, in Section \ref{main_section} it will become apparent that this intermediate generalised resolvent
itself is, up to the same correcting rank-one operator,
 $O(\e^2)$-close in the operator-norm sense to the resolvent of a unitary transformation of a self-adjoint operator $A_{\text{hom}}$, yielding the estimate \eqref{main_est}.

We first describe a modification procedure for the original cycle graph ${\mathbb G},$
see Fig.\,\ref{fig:mod}.
\begin{figure}
\begin{center}
\includegraphics[width=.8\textwidth]{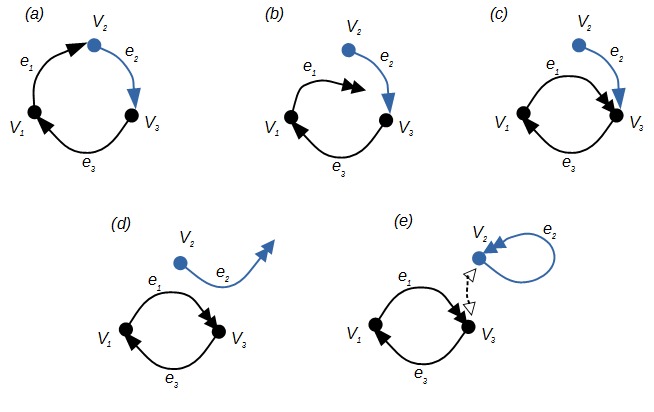}
\end{center}
\caption{{\scshape The graph modification.} {\small The stiff component is in black, the soft component is in blue. Double arrows represent vertices carrying unimodular Datta--Das Sarma weights. Dotted arrowed line between vertices $V_2$ and $V_3$ represents non-local $\e-$dependent interface.}}\label{fig:mod}
\end{figure}
The modified graph $\widetilde{\mathbb G}$ is a two-component graph with
edges
$e_1
\equiv \tilde e_1
:=[0,\e l_1]$,
$e_3\equiv \tilde
e_3:=[0, \e l_3],$ and $e_2:=[0,l_2].$
The edges $e_1$ and $e_3$ are ``glued'' together, forming a cycle with two vertices. Compared to the original graph ${\mathbb G}$ (Fig.\,\ref{fig:mod}(a)), the vertex $V_1$ remains unchanged, whereas the right endpoint $\e l_1$ of the edge $e_1$ disconnects from the vertex $V_2$ (Fig.\,\ref{fig:mod}(b)) and joins
$V_3,$ which is the left endpoint of $e_3$
(Fig.\,\ref{fig:mod}(c)). There is a price to be paid for this: this right endpoint of $e_1$ is then assigned a Datta -- Das Sarma unimodular weight $w_{\rm stiff}:=\exp\bigl({\rm i}(l_1+l_3)\tau\bigr).$
The edge $e_2$ in turn disconnects from the vertex $V_3$ where its right endpoint was attached to in
${\mathbb G}$ (Fig.\,\ref{fig:mod}(d)), and loops backwards to the vertex $V_2$ (Fig.\,\ref{fig:mod}(e)). The loop thus formed is assigned a Datta -- Das Sarma 
weight
\[
w_{\rm soft}=\overline{w_{\rm stiff}}=\exp\bigl(-{\rm i}(l_1+l_3)\tau\bigr)
\]
 at its right endpoint $l_2$. Compared to the graph ${\mathbb G}$, the  weights $1/\sqrt{\e}$ at both endpoints of the weak component are no longer applied.
 Notice also that
the weights $w_{\rm stiff},$ $w_{\rm soft}$ are independent of $\e,$ which is important in view of our aim to obtain an
$\varepsilon$-independent family $A_{\rm hom}^{(\tau)}$ in the estimate (\ref{main_est}).

The operator $\tilde A_{\e}^{(t)}$ is defined by the same differential expression as the operator $\Phi_\varepsilon A_\e^{(t)}\Phi_\varepsilon^*,$ which 
has an $\e$-independent form. The domain of $\tilde A_\e^{(t)},$ however, depends on $\e$ as well as on $k^2$ and is described by the following system of matching conditions (\ref{matching1})--(\ref{matching4}). We always assume $u=(u_1,u_2,u_3)$ with respect to the space decomposition, where $u_2$ is the value on the soft component.

A.
At the vertex $V_1$: standard $\delta$-type matching with the coupling constant equal to zero.

B.
At the vertex $V_3$ (stiff component):
\begin{equation}
u_3(0)=w_{\rm stiff} u_1(\e l_1),\ \ \ \ \ \ \ \  \ \ \ \ \ \ \ \  \ \ \ \ \ \ \ \  \ \ \ \ \ \ \ \  \ \ \ \ \ \ \ \  \ \ \ \ \ \ \ \  \ \ \ \ \ \ \ \  \ \ \ \ \ \ \ \ \ \  \ \ \ \ \ \ \ \  \ \ \ \ \ \ \ \  \ \ \ \ \ \
\label{matching1}
\end{equation}
\vskip -0.55cm
\begin{equation}
\partial^{(t)}u_3\bigr\vert_0-w_{\rm stiff}\partial^{(t)}u_1\bigr\vert_{\e l_1}=\sqrt{\e} k^2(l_1+l_3)w_{\rm stiff} u_2(0).\ \ \ \ \ \ \ \  \ \ \ \ \ \ \ \  \ \ \ \ \ \ \ \  \ \ \ \ \ \ \ \ \ \ \ \ \ \ \ \ \ \ \ \ \
\label{matching2}
\end{equation}

C.
At the vertex $V_2$ (soft component):
\begin{equation}
u_2(0)=w_{\rm soft} u_2(l_2), \ \ \ \  \ \ \ \ \ \ \ \  \ \ \ \ \ \ \ \ \ \ \ \ \ \  \ \ \ \ \ \ \ \  \ \ \ \ \ \ \ \  \ \ \ \ \ \ \ \  \ \
\ \ \ \  \ \ \ \ \ \ \ \  \ \ \ \ \ \ \ \ \ \ \ \ \ \ \ \ \ \ \ \ \ \ \ \
\label{matching3}
\end{equation}
\vskip -0.55cm
\begin{equation}
\partial^{(\tau)}u_2\bigr\vert_0-w_{\rm soft}\partial^{(\tau)}u_2\bigr\vert_{l_2}=\sqrt{\e} k^2(l_1+l_3)w_{\rm soft} u_3(0)
-2 k^2 (l_1+l_3) u_2(0).\ \ \ \ \ \ \ \ \ \ \ \ \ \ \ \ \ \ \ \ \
\label{matching4}
\end{equation}
Clearly, all these conditions are of $\delta$-type, with
$\e$-dependent non-local terms in (\ref{matching2}) and (\ref{matching4}), which link the two components.

The operator $\tilde A_\e^{(t)}$ is written down in terms of the Datta -- Das Sarma boundary triple,
see Section \ref{Datta_section}, for the modified graph
$\widetilde{\mathbb G}.$ It involves Datta -- Das Sarma matching conditions at two of the three graph vertices, namely, $V_2$ (incoming edge endpoint, weight $w_{\rm soft}$) and $V_3$ (incoming edge endpoint, weight $w_{\rm stiff}$).  We denote by $\tilde{\Gamma}_0^{(2)},$ $\tilde{\Gamma}_1^{(2)}$ the corresponding boundary operators and by
$\tilde{B}^{(2)}(z)$ the matrix such that the interface conditions (\ref{matching1})--(\ref{matching4}) are equivalent to
$$\tilde{\Gamma}_1^{(2)}u=\tilde{B}^{(2)}(z)\tilde{\Gamma}_0^{(2)}u.$$
Omitting the details of the calculation for $\tilde{B}^{(2)}(z)$ and for the $M$-matrix $\tilde{M}^{(2)}_\varepsilon(z)$ of the operator
$\tilde A_\e^{(t)}$ with respect to $\tilde{\Gamma}_0^{(2)},$ $\tilde{\Gamma}_1^{(2)}$ (which is analogous to the calculations of Sections \ref{Mmatrixsection} and \ref{section_rescaling}), we claim that
$$
\tilde{M}^{(2)}_\varepsilon(z)-\tilde{B}^{(2)}(z)=\left(\begin{array}{ccc}
-\sqrt{a}k\sum\limits_{j=1,3} \cot\dfrac{k l_j \e}{\sqrt{a}}&0
\\[1.3em]
0& \dfrac{2k(\cos \varepsilon t-\cos k l_2)}{\sin k l_2}+2 k^2(l_1+l_3)
\\[0.6em] {\rm e}^{{\rm i}\e l_3 t}\sqrt{a}k\sum\limits_{j=1,3}\biggl(\sin\dfrac{k l_j \e}{\sqrt{a}}\biggr)^{-1}&-{\rm e}^{{\rm i}\e(l_1+l_3)t}k^2 (l_1+l_3)\sqrt{\e}
\end{array}\right.\ \ \ \ \ \ \ \ \ \ \ \ \ \ \ \ \ \ \ \ \ \ \ \ \ \ \ \ \ \ \ \ \
$$
\begin{equation}
\ \ \ \ \ \ \ \ \ \ \ \ \ \ \ \ \ \ \ \ \ \ \ \ \ \ \ \ \ \ \ \ \ \ \ \ \ \ \ \ \ \ \ \ \ \ \ \ \ \ \ \ \ \ \ \ \ \ \ \ \ \ \ \ \ \ \ \ \ \ \ \ \ \ \ \ \ \ \ \ \left.\begin{array}{c}
{\rm e}^{-{\rm i}\e l_3 t}\sqrt{a}k\sum\limits_{j=1,3}\biggl(\sin\dfrac{k l_j \e}{\sqrt{a}}\biggr)^{-1}\\[1.5em] -{\rm e}^{-{\rm i}\e(l_1+l_3)t}k^2 (l_1+l_3)\sqrt{\e}\\
[0.7em]-\sqrt{a}k\sum\limits_{j=1,3} \cot\dfrac{k l_j \e}{\sqrt{a}}
\end{array}\right).
\label{MBtilde_diff}
\end{equation}
We
argue that the difference between the resolvent of
$\Phi_\varepsilon A_\e^{(t)}\Phi_\varepsilon^*$ and the generalised resolvent of $\tilde A_\e^{(t)}$
is
of order $O(\e^2)$ in the operator-norm sense, up to a
``correcting" operator, which
 takes into account the
difference between the kernels of the
$\Phi_\varepsilon A_\e^{(t)}\Phi_\varepsilon^*$ and $\tilde A_\e^{(t)}$ on the stiff component and is
$O(\e^2)$-close to a rank-one operator multiplied by $z^{-1}.$
Once the mentioned estimate is obtained,
 it is possible to
 eliminate $\e$  from the domain description of the operator $\tilde A_\e^{(t)},$ which
 can therefore be viewed as intermediate from the perspective of homogenisation. 
We keep this
step
explicit, owing to the fact that the resolvent estimate in this form does not require the assumption that the spectral parameter belongs to a compact set. It therefore shows what happens during the transition from the ``classical'' homogenisation regime to the ``high-frequency''  regime, when
 the norm of the correcting rank-one operator discussed above goes to zero as $\e\to0.$
 In the present paper
we refrain from
discussing the related details and
assume that the spectral parameter $z$ belongs to a compact set $K\subset \mathbb C$. We point out that in the
transition regime
the error estimates in the statements given at the end of the present section are changed accordingly, which will be studied
elsewhere.

In order that the Kre\u\i n formula of Section \ref{sect:Krein} be applicable, we must
ensure that the spectral parameter is away from the zeroes of the denominator. Let $S_{\text{hom}}^{(t)}$ be the limiting spectrum of the family $A_\e^{(t)}$ described by \eqref{CC}, and let $S_\infty$ be the set of eigenvalues of the Dirichlet boundary-value problem of the operator $-d^2/dx^2$ on the soft component $e_2,$
\emph{i.e.} the set of points $z>0$ such that $\sin \sqrt{z} l_2=0$. Setting ({\it cf.} (\ref{z_set}))
\begin{equation}\label{eq:Srho}
S^{(t)}:=S_{\text{hom}}^{(t)}\cup S_\infty \cup \{0\}, \quad \quad S_{K,\rho}^{(t)}:=\bigl\{z\in K\,|\,\text{dist}\,\bigl(z,S^{(t)}\bigr)\geq \rho>0\bigr\},
\end{equation}
the following theorem holds.

\begin{theorem}
\label{comparison_theorem}
Denote
\begin{equation}
{\mathcal X}^{(t)}(x):=\begin{cases}{\rm e}^{-{\rm i}tx}, & x\in e_1,\\[0.2em]
                              {\rm e}^{{\rm i}t(\e l_3-x)}, & x\in e_3,
                              \end{cases}
\label{Chi_capital}
\end{equation}
\begin{equation}
L_{\rm stiff}^2:=L^2(e_1)\oplus L^2(e_3),
\label{L2s_def}
\end{equation}
and consider the $z$-dependent linear operator $C^{(t)}$ on $H$ given by
\begin{equation}\label{Ct_def1}
C^{(t)} \begin{pmatrix} f_1 \\ f_2 \\ f_3 \end{pmatrix}:=
\begin{pmatrix}
P_{e_1}
\\[0.2em]
0\\[0.2em]
P_{e_3}
\end{pmatrix}\widehat C^{(t)}\begin{pmatrix} f_1 \\ f_3 \end{pmatrix},
\ \ \ \ \ \ \ \ \ \widehat C^{(t)}[\cdot]:=\bigl(\e z(l_1+l_3)\bigr)^{-1}\bigl\langle \cdot, {\mathcal X}^{(t)}\bigr\rangle_{L^2_{\rm stiff}}{\mathcal X}^{(t)},
\end{equation}
where $P_{e_j}$
is the orthogonal protection onto $L^2(e_j),$
$j=1,3.$
Then the following estimate holds:
\begin{equation}\label{eq:intermediate}
\bigl\|\Phi_\varepsilon \bigl(A_\e^{(t)}-z\bigr)^{-1}\Phi_\varepsilon^*-\bigl(\tilde A_\e^{(t)} -z\bigr)^{-1}-C^{(t)}\bigr\|=O(\e^2),
\end{equation}
uniformly with respect to $t\in[0,2\pi\varepsilon^{-1})$ for all $z\in S_{K,\rho}^{(t)},$ and therefore, as is seen from the explicit expression for $\bigl(M^{(2)}_\varepsilon(z)-\tilde{B}^{(2)}(z)\bigr)^{-1}$ below, away from the set of singularities of the generalised resolvent $\bigl(\tilde A_\e^{(t)}-z\bigr)^{-1}$.
\end{theorem}

\begin{proof}

We start with the following lemma.

\begin{lemma}
\label{asymptotic_lemma}
1. The inverse of the matrix $M^{(2)}_\varepsilon(z),$ see (\ref{M2form}), has the following asymptotics as $\e\to 0$:
$$
\bigl(M^{(2)}_\varepsilon(z)\bigr)^{-1}=\frac{1}{\e D(k)}
\begin{pmatrix}
1& {\rm e}^{{\rm i}\e l_1 t} & {\rm e}^{-{\rm i}\e l_3 t}\\[0.4em]
{\rm e}^{-{\rm i}\e l_1 t} & 1 & {\rm e}^{-{\rm i}\e (l_1+l_3)t}\\[0.4em]
{\rm e}^{{\rm i}\e l_3 t}& {\rm e}^{{\rm i}\e (l_1+l_3)t}&1
\end{pmatrix}+ O(\e),
$$
uniformly with respect to $t\in[0,2\pi\varepsilon^{-1})$ for all $z\in S_{K,\rho}^{(t)}$, where
\[
D(k):=k^2(l_1+l_3)-2k\cot k l_2 + \frac{2k\cos t}{\sin k l_2},
\]
and the matrix defining the leading term of order  $O(1/\e)$  is a rank-one matrix with the range spanned by the eigenvector $\bigl(e^{-{\rm i}\e l_3 t}, e^{-{\rm i} \e (l_1+l_3)t},1\bigr)^\top$ corresponding to the eigenvalue 3.


2. For the inverse of the matrix $\tilde M^{(2)}_\varepsilon(z)-\tilde{B}^{(2)}(z),$ see (\ref{MBtilde_diff}), one has
$$
\bigl(\tilde M^{(2)}_\varepsilon(z)-\tilde{B}^{(2)}(z)\bigr)^{-1}=
\varepsilon^{-1}M_{-1}(z) + M_0(z)+ \e M_1(z)+O\bigl(\e^2\bigr)M_0(z)+O(\e^2),\ \ \ \varepsilon\to0,
$$
uniformly with respect to $t\in[0,2\pi\varepsilon^{-1})$ for all $z\in S_{K,\rho}^{(t)}$, where
$$M_{-1}(z):=\left(\frac 1 {k^2(l_1+l_3)}+\frac 1{D(k)} \right)
\begin{pmatrix}
   1&0&{\rm e}^{-{\rm i}\e l_3 t}\\[0.3em]
   0&0&0\\[0.3em]
   {\rm e}^{{\rm i}\e l_3 t}&0&1
   \end{pmatrix},
$$
$$
M_0(z):= \frac{1}{D(k)}
            \begin{pmatrix}
            0&\dfrac{1}{\sqrt{\e}}{\rm e}^{{\rm i}\e l_1 t}&0\\[0.3em]
            \dfrac{1}{\sqrt{\e}}{\rm e}^{-{\rm i}\e l_1 t}&1&\dfrac{1}{\sqrt{\e}}{\rm e}^{-{\rm i}\e (l_1+l_3)t}
            \\[0.3em] 0&\dfrac{1}{\sqrt{\e}}{\rm e}^{{\rm i}\e (l_1+l_3) t}&0
            \end{pmatrix},
$$
and the matrix $M_1(z)$ has
all but the four corner elements vanishing.

\end{lemma}

\begin{proof} This is the result of a direct
calculation.
\end{proof}

In order to compare the two resolvents, we use the Kre\u\i n resolvent formula of Proposition \ref{prop:Krein} twice, namely for $\Phi_\varepsilon A_\varepsilon^{(t)}\Phi_\varepsilon^*$ and $\tilde A_\e^{(t)},$ as well as the observation that
in both cases
the ``reference operator'' $A_\infty$ is
the same Dirichlet decoupling: on each edge $e_j$ of both ${\mathbb G}$ and $\widetilde{\mathbb G}$ it is the differential operator defined by the corresponding differential expression subject to Dirichlet conditions at both endpoints, $u_j(0)=u_j(\e l_j)=0$ for $j=1,3,$  or $u_2(0)=u_2(l_2)=0$. Note that the
operator $B,$ see Definition \ref{definition1_1}, for
$A_\e^{(t)}$ with respect to the triple of
Section \ref{section_rescaling} is the zero matrix, and hence the matrix $-M^{(2)}_\varepsilon(z)$ plays the role of
the operator $B(z)-M(z)$ in the corresponding Kre\u\i n formula.

We consider three cases for the form of the argument of the resolvents, as follows.

I. First, we apply the two mentioned resolvents to functions
$f=(0,f_2,0)^\top$. Then $(A_\infty - z)^{-1}f=(0,v,0)^\top$ and
\begin{gather*}
\Gamma_1^{(2)}\begin{pmatrix} 0\\ v\\ 0\end{pmatrix} =
\sqrt{\e} \begin{pmatrix} 0\\ \partial^{(\tau)}_nv\bigr|_0 \\ \partial^{(\tau)}_n v\bigr|_{l_2} \end{pmatrix}=:\sqrt{\e} \begin{pmatrix} 0\\ \alpha_2 \\ \beta_2 \end{pmatrix},\\
\tilde \Gamma_1^{(2)}\begin{pmatrix} 0\\ v\\ 0 \end{pmatrix}=
\begin{pmatrix}
0\\
\dntau v\bigr|_0+{\rm e}^{-{\rm i}\e (l_1+l_3)t}\dntau v\bigr|_{l_2}\\
0
\end{pmatrix}=
\begin{pmatrix}
0\\
\alpha_2 +{\rm e}^{-{\rm i}\e(l_1+l_3)t}\beta_2\\
0
\end{pmatrix}=:
\begin{pmatrix}
0\\ \gamma_2 \\ 0
\end{pmatrix}.
\end{gather*}
Using Lemma \ref{asymptotic_lemma}, we obtain:
\begin{gather*}
\bigl(M^{(2)}_\varepsilon(z)\bigr)^{-1}\Gamma_1^{(2)} \begin{pmatrix}0\\ v \\ 0 \end{pmatrix}=
\frac{1}{\sqrt{\e}D(k)}
\begin{pmatrix}
\gamma_2{\rm e}^{{\rm i}\e l_1 t} \\[0.3em]
\gamma_2 \\[0.3em]
\gamma_2{\rm e}^{{\rm i}\e (l_1+l_3)t}
\end{pmatrix}
+\begin{pmatrix}
O(\varepsilon^{3/2})\\[0.3em]
O(\varepsilon^{3/2})\\[0.3em]
O(\varepsilon^{3/2})
\end{pmatrix},\\
\bigl(\tilde{M}^{(2)}_\varepsilon(z)-\tilde{B}^{(2)}(z)\bigr)^{-1}\tilde{\Gamma}_1^{(2)}\begin{pmatrix}0\\ v \\ 0 \end{pmatrix}=
M_0\tilde{\Gamma}_1^{(2)}\begin{pmatrix}0\\ v \\ 0 \end{pmatrix} +\dots =
\frac{1}{D(k)}
\begin{pmatrix}
\dfrac{\gamma_2}{\sqrt{\e}}{\rm e}^{{\rm i}\e l_1 t} \\[0.3em]
\gamma_2 \\[0.3em]
\dfrac{\gamma_2}{\sqrt{\e}}{\rm e}^{{\rm i}\e (l_1+l_3)t}
\end{pmatrix}+
\begin{pmatrix}
O(\varepsilon^{3/2})\\[0.3em]
O(\varepsilon^2)\\[0.3em]
O(\varepsilon^{3/2})
\end{pmatrix}.
\end{gather*}
It remains to apply the solution operators $\gamma(z)$ and $\tilde\gamma(z)$ of Proposition \ref{prop:Krein},  pertaining to the boundary triples of operator families $A_\e^{(t)}$ and $\tilde A_\e^{(t)}$, respectively. This amounts to comparing solutions to three pairs of boundary-value problems, on $e_1$, $e_2,$ and $e_3.$

(a) Solutions on $e_2.$
Due to the definitions of boundary triples, to the leading order in each case one solves boundary-value problems  with the boundary data
\[
u_2(0)=\frac{\gamma_2}{D(k)},\ \ \ \ \ u_2(l_2)=\frac{\gamma_2}{D(k)}{\rm e}^{{\rm i}\e(l_1+l_3)t},
\]
with an error of order $O(\varepsilon^2)$ between the contributions to the resolvents $\bigl(A_\varepsilon^{(t)}-z\bigr)^{-1}$ and $\bigl(\tilde A_\varepsilon^{(t)}-z\bigr)^{-1}.$

(b) Solutions on $e_1.$
In both cases, to the leading order one gets the solution to the boundary-value problems with the data
\[
u_1(0)=\frac{\gamma_2}{\sqrt{\e}D(k)}{\rm e}^{{\rm i}\e l_1 t} ,\ \ \ \ \ \
u_1(\e l_1)=\frac{\gamma_2}{\sqrt{\e}D(k)}.
\]

(c) In the case of $e_3,$
to the leading order one also gets the same solution for both $A_\e^{(t)}$ and $\tilde A_\e^{(t)},$ which is fixed by
\[
u_3(0)=\frac{\gamma_2}{\sqrt{\e}D(k)}{\rm e}^{{\rm i}\e (l_1+l_3) t} ,\ \ \ \
u_3(\e l_3)=\frac{\gamma_2}{\sqrt{\e} D(k)}{\rm e}^{{\rm i}\e l_1 t}.
\]
In the cases (b), (c)
(stiff component), the error between the actions of
the resolvents
$\bigl(A_\varepsilon^{(t)}-z\bigr)^{-1}$ and $\bigl(\tilde A_\varepsilon^{(t)}-z\bigr)^{-1}$ is of the order
$O(\varepsilon^2).$ Indeed, the pointwise error is of the order $O(\varepsilon^{3/2}),$ and $e_1,$ $e_3$ have lengths proportional to $\e$.

II. Now consider vectors $f=(f_1,0,0)^\top$. Denoting
$(A_\infty - z)^{-1}f=(v,0,0)^\top,$ one has
\begin{equation*}
\Gamma_1^{(2)}\begin{pmatrix} v\\0 \\ 0\end{pmatrix} =
 \begin{pmatrix} \dnt v\bigr|_0\\ \dnt v\bigr|_{\e l_1}\\ 0 \end{pmatrix}=: \begin{pmatrix} \alpha_1 \\ \beta_1 \\0 \end{pmatrix},
 \quad
\tilde \Gamma_1^{(2)} \begin{pmatrix} v\\ 0 \\ 0\end{pmatrix}=
\begin{pmatrix}
\dnt v\bigr|_0\\
0\\
{\rm e}^{{\rm i}\e (l_1+l_3)t} \dnt v\bigr|_{\e l_1}
\end{pmatrix}=
\begin{pmatrix}
\alpha_1\\
0\\
{\rm e}^{{\rm i}\e (l_1+l_3)t} \beta_1
\end{pmatrix}.
\end{equation*}
Denoting $\gamma_1:=\alpha_1 + \beta_1{\rm e}^{{\rm i}\e l_1t}$ and using Lemma \ref{asymptotic_lemma} again, we obtain:
\begin{gather*}
\bigl(M^{(2)}_\varepsilon(z)\bigr)^{-1}\Gamma_1^{(2)}\begin{pmatrix}v\\ 0 \\ 0 \end{pmatrix}=
\frac{1}{\e D(k)}
\begin{pmatrix}
\alpha_1 + \beta_1{\rm e}^{{\rm i}\e l_1t} \\[0.4em]
\alpha_1{\rm e}^{-{\rm i}\e l_1t} +\beta_1 \\[0.4em]
\alpha_1{\rm e}^{{\rm i}\e l_3 t} + \beta_1{\rm e}^{{\rm i}\e (l_1+l_3)t}
\end{pmatrix}
+\dots
=\frac{1}{\e D(k)}
\begin{pmatrix}
\gamma_1 \\[0.3em] \gamma_1{\rm e}^{-{\rm i}\e l_1 t}\\[0.3em] \gamma_1{\rm e}^{{\rm i} \e l_3 t}
\end{pmatrix}+\gamma_1 O(\e),
\\
\bigl(\tilde{M}^{(2)}_\varepsilon(z)-\tilde{B}^{(2)}(z)\bigr)^{-1}\tilde{\Gamma}_1^{(2)}\begin{pmatrix}v\\ 0 \\ 0 \end{pmatrix}=
\frac 1\e \left(\frac 1{k^2(l_1+l_3)}+ \frac{1}{D(k)} \right)
\begin{pmatrix}
 \gamma_1 \\[0.3em]
0 \\[0.3em]
\gamma_1{\rm e}^{{\rm i}\e l_3t}
\end{pmatrix}
+\frac 1{D(k)}
\begin{pmatrix}
0
\\[0.3em]
\dfrac{\gamma_1}{\sqrt{\e}}{\rm e}^{-{\rm i}\e l_1 t}\\[0.3em]
0
\end{pmatrix}
+\gamma_1 \begin{pmatrix}
O(\e)\\[0.4em]
O(\e^{3/2})\\[0.35em]
O(\e)
\end{pmatrix}.
\end{gather*}
In contrast to $\gamma_2$ in the case considered above, the coefficient $\gamma_1$ is
of the order $O(\sqrt{\e})$ rather than $O(1)$. Indeed,
the operator
$(A_\infty - z)^{-1}$
on $L^2(e_1)$
is simply
the resolvent of the self-adjoint Dirichlet operator $L_D$ defined by the differential expression
\[
a\biggl(\dfrac {1}{\rm i}\frac{d}{dx}+t\biggr)^2
\]
on
$L^2(e_1)$. It is an integral operator with a kernel $R(x,y;k)$ that can be found by the classical method of \cite{Sargsyan}, \cite{Marchenko} combined with the unitary elimination of the ``magnetic potential'' $t\sqrt{a}.$ Namely, let $A_D$ be the Dirichlet operator on the same space defined by the expression $-a({d^2}/{dx^2}),$
and let $\Phi$ be the unitary transformation
$(\Phi u)(x)={\rm e}^{-{\rm i}tx}u(x)$. Then
$L_D=\Phi A_D \Phi^*,$ and hence $(L_D-z)^{-1}=\Phi (A_D-z)^{-1} \Phi^*.$ The resolvent of $A_D$ is well-known, see {\it e.g.} \cite{Sargsyan}: it is the integral operator with kernel
$$
R_A(x,y;k)=\biggl(\sqrt{a}k\sin\dfrac{k \e l_1}{\sqrt{a}}\biggr)^{-1}\begin{cases}
\sin\dfrac{kx}{\sqrt{a}}\sin\dfrac {k(\e l_1-y)}{\sqrt{a}}, & x<y,\\[0.9em]
\sin\dfrac{ky}{\sqrt{a}}\sin\dfrac {k(\e l_1-x)}{\sqrt{a}}, & x>y.
\end{cases}
$$
Using the fact that $R(x,y;k)={\rm e}^{-{\rm i}tx}R_A(x,y;k){\rm e}^{{\rm i}ty},$ it follows that
\begin{multline*}
\dt\bigl((L_D-z)^{-1}f\bigr)(x)=
-{\rm e}^{-{\rm i}tx}\biggl(\sin \frac{k\e l_1}{\sqrt{a}}\biggr)^{-1}
\left[\cos\frac{k(\e l_1-x)}{\sqrt{a}}\int_0^x \sin\frac{ky}{\sqrt{a}}{\rm e}^{{\rm i}ty}f(y)dy\right.\\
-\left.\cos\frac{k x}{\sqrt{a}}\int_x^{\e l_1} \sin\frac{k(\e l_1-y)}{\sqrt{a}}{\rm e}^{{\rm i}ty}f(y)dy
\right].
\end{multline*}
Substituting trigonometric functions by the leading-order terms, as $\varepsilon\to0,$ of their power series
yields
\[
\dt\bigl((L_D-z)^{-1}f\bigr)(x)=
-\frac{1}{\e l_1}{\rm e}^{-{\rm i}tx}\left[\int_0^{\e l_1} y{\rm e}^{{\rm i}ty} f(y) dy -\e l_1 \int_x^{\e l_1}{\rm e}^{{\rm i}ty} f(y) dy\right]\bigl(1+O(\e^2)\bigr)+O(\e^{5/2})\|f\|,
\]
and therefore
\begin{equation*}
\gamma_1=\dt u\bigr|_0-e^{{\rm i}\e l_1 t}\dt u\bigr|_{\e l_1}=\int_0^{\e l_1}{\rm e}^{{\rm i}ty} f(y) dy\bigl(1+O(\e^2)\bigr)+O(\e^{5/2})\|f\|\ \ \ \ \ \ \ \ \ \ \ \ \ \ \ \ \ \ \ \ \ \ \ \ \ 
\end{equation*}
\begin{equation}
\label{eq:gamma}
\ \ \ \ \ \ \ \ \ \ \ \ \ \ \ \ \ \ \ \ \ \ \ \ \ \ \ \ \ \ \ \ \ \ \ \ \ \ \ \ \ \ \ \ \ \ \ \ \ \ =\bigl\langle f,{\rm e}^{-{\rm i}ty}\bigr\rangle_{L^2(e_1)}\bigl(1+O(\e^2)\bigr)+O(\e^{5/2})\|f\|.
\end{equation}
Notice that by the Kre\u\i n resolvent formula the term $O(\e^{5/2})\|f\|$ contributes an error of order $O(\e^2)$ in the resolvent estimate and can therefore be discarded. An application of the Schwartz inequality
yields $\gamma_1=O(\sqrt{\e})$, as claimed. It again remains to apply the operators $\gamma(z)$ and $\tilde\gamma(z).$ 

(a) Solutions on $e_2.$
Due to the definitions of the boundary triples, to the leading order in each case one solves boundary-value problems with boundary data
\[
u_2(0)=\frac{\gamma_1}{\sqrt{\e}D(k)}{\rm e}^{-{\rm i}\e l_1 t},\ \ \ \ \
u_2(l_2)=\frac{\gamma_1 }{\sqrt{\e}D(k)}{\rm e}^{{\rm i}\e l_3t}.
\]
with an error of order $O(\varepsilon^2)$ between the contributions to the resolvents $\bigl(A_\varepsilon^{(t)}-z\bigr)^{-1}$ and $\bigl(\tilde A_\varepsilon^{(t)}-z\bigr)^{-1}.$

(b) Solutions on $e_1$
In the case of $\tilde A_\e^{(t)}$, to the leading order one solves the boundary-value problem
with data
\[
u_1(0)=\frac{\gamma_1}{\e}\left( \frac 1 {k^2 (l_1+l_3)} +\frac 1 {D(k)} \right),\ \ \ \ \
u_1(\e l_1)=\frac{\gamma_1}{\e}\left( \frac 1 {k^2 (l_1+l_3)} +\frac 1 {D(k)} \right){\rm e}^{-{\rm i} \e l_1 t},
\]
whereas in the case of $A_\e^{(t)}$ the boundary data to the leading order are
\[
u_1(0)=\frac{\gamma_1}{\e D(k)},\ \ \ \ \ \ \ \ \ \
u_1(\e l_1)=\frac{\gamma_1}{\e D(k)}{\rm e}^{-{\rm i} \e l_1 t}.
\]
Clearly, a correcting boundary-value problem appears, for the
``stiff component to stiff component'' action of the intermediate generalised resolvent only.

(c) Solutions on $e_3.$
As in (b) above, a correcting boundary-value problem appears, which has the same form. Indeed, in the case of $\tilde A_\e^{(t)}$, to the leading order one solves the boundary-value problem
with boundary data
$$
u_3(0)=\frac{\gamma_1}{\e}\left( \frac 1 {k^2 (l_1+l_3)} +\frac 1 {D(k)} \right){\rm e}^{{\rm i} \e l_3 t},\ \ \ \ \ \ \
u_3(\e l_3)=\frac{\gamma_1}{\e}\left( \frac 1 {k^2 (l_1+l_3)} +\frac 1 {D(k)} \right),
$$
whereas in the case of $A_\e^{(t)}$ one has
\[
u_3(0)=\frac{\gamma_1}{\e D(k)}{\rm e}^{{\rm i}\e l_3 t},\ \ \ \ \ \ \ \
u_3(\e l_3)=\frac{\gamma_1}{\e D(k)}.
\]
In the cases (b), (c),
the error between the actions
the resolvents $\Phi_\e (A_\varepsilon^{(t)}-z)^{-1} \Phi_\e^*$ and $(\tilde A_\varepsilon^{(t)}-z)^{-1},$ up to the correcting term mentioned above, is of the order
$O(\varepsilon^2),$ due to the pointwise error being of the order $O(\varepsilon^{3/2}).$ Here we again use the fact that $e_1$ and $e_3$ have lengths proportional to $\e,$ as well as the above estimate for $\gamma_1.$

III. Finally, in the case $f=(0,0,f_3)^\top$ one has
\begin{equation*}
\Gamma_1^{(2)}\begin{pmatrix} 0\\0 \\ v\end{pmatrix} =
 \begin{pmatrix} \dnt v\bigr|_{\e l_3}\\ 0 \\ \dnt v\bigr|_0 \end{pmatrix}=: \begin{pmatrix} \beta_3 \\ 0 \\ \alpha_3 \end{pmatrix},\ \ \ \ \ \
\tilde \Gamma_1^{(2)}\begin{pmatrix} 0\\ 0 \\ v\end{pmatrix}=
\begin{pmatrix}
\dnt v\bigr|_{\e l_3}\\
0\\
 \dnt v\bigr|_0
\end{pmatrix}=
\begin{pmatrix}
\beta_3\\
0\\
\alpha_3
\end{pmatrix},
\end{equation*}
where we set $(A_\infty - z)^{-1}f=:(0,0,v)^\top.$
Denoting $\gamma_3:=\beta_3 + \alpha_3{\rm e}^{-{\rm i}\e l_3t}$ and using Lemma \ref{asymptotic_lemma}, we obtain:
\begin{gather*}
\bigl(M^{(2)}_\varepsilon(z)\bigr)^{-1}\Gamma_1^{(2)}\begin{pmatrix}0\\ 0 \\ v \end{pmatrix}=
\frac{1}{{\e}D(k)}
\begin{pmatrix}
\beta_3 + \alpha_3{\rm e}^{-{\rm i}\e l_3t} \\[0.4em]
\beta_3{\rm e}^{-{\rm i}\e l_1t} +\alpha_3{\rm e}^{-{\rm i}\e(l_1+l_3)t} \\[0.4em]
\beta_3{\rm e}^{{\rm i}\e l_3 t}+\alpha_3
\end{pmatrix}
+\dots
=\frac 1 {\e D(k)}\begin{pmatrix}
\gamma_3 \\[0.4em] \gamma_3{\rm e}^{-{\rm i}\e l_1 t} \\[0.4em] \gamma_3{\rm e}^{{\rm i} \e l_3 t}
\end{pmatrix}+\gamma_3 O(\e),
\\
\bigl(\tilde{M}^{(2)}_\varepsilon(z)-\tilde{B}^{(2)}(z)\bigr)^{-1}\tilde{\Gamma}_1^{(2)}\begin{pmatrix}0\\ 0 \\ v \end{pmatrix}=
\frac 1\e \left(\frac 1{k^2(l_1+l_3)}+ \frac{1}{D(k)} \right)
\begin{pmatrix}
 \gamma_3 \\[0.4em]
0 \\[0.4em]
\gamma_3 {\rm e}^{{\rm i}\e l_3t}
\end{pmatrix}
+\frac 1{D(k)}
\begin{pmatrix}
0\\[0.4em]
\dfrac{\gamma_3}{\sqrt{\e}}{\rm e}^{-{\rm i}\e l_1 t}\\[0.4em]
0
\end{pmatrix}
+\gamma_3 \begin{pmatrix}
O(\e)\\
O(\e^{3/2})\\
O(\e)
\end{pmatrix}.
\end{gather*}
An argument similar to the case of $\gamma_1$
yields the estimate $\gamma_3=O(\sqrt{\e})$. We now apply the operators $\gamma(z)$ and $\tilde\gamma(z).$

(a) Solutions on $e_2.$
Due to the definitions of the boundary triples, in both cases to the leading order one solves
the boundary-value problem with data
\[
u_2(0)=\frac{\gamma_3}{\sqrt{\e} D(k)}e^{-{\rm i}\e l_1 t},\ \ \ \ \ \ u_2(l_2)=\frac{\gamma_3}{\sqrt{\e}D(k)}{\rm e}^{{\rm i}\e l_3t}
\]
yielding an error of order $O(\varepsilon^2)$ between the actions of
the resolvents $\Phi_\e(A_\varepsilon^{(t)}-z)^{-1}\Phi_\e^*$ and $(\tilde A_\varepsilon^{(t)}-z)^{-1}.$

(b) Solutions on $e_1$
In the case of $\tilde A_\e^{(t)}$, to the leading order one solves the boundary-value problem
with
data
$$
u_1(0)=\frac{\gamma_3}{\e}\left( \frac 1 {k^2 (l_1+l_3)} +\frac 1 {D(k)} \right), \ \ \ \
u_1(\e l_1)=\frac {\gamma_3}\e \left( \frac 1 {k^2 (l_1+l_3)} +\frac 1 {D(k)} \right){\rm e}^{-{\rm i} \e l_1 t},
$$
whereas in the case of $A_\e^{(t)}$ one has
\[
u_1(0)=\frac{\gamma_3}{\e D(k)},\ \ \ \ \
u_1(\e l_1)=\frac{\gamma_3}{\e D(k)}{\rm e}^{-{\rm i} \e l_1 t}.
\]

(c) Solutions on $e_3.$
In the case of $\tilde A_\e^{(t)}$, to the leading order one solves the boundary-value problem  with 
data
$$
u_3(0)=\frac{\gamma_3}{\e}\left( \frac 1 {k^2 (l_1+l_3)} +\frac 1 {D(k)} \right){\rm e}^{{\rm i} \e l_3 t},\ \ \ \
u_3(\e l_3)=\frac{\gamma_3}{\e}\left( \frac 1 {k^2 (l_1+l_3)} +\frac 1 {D(k)} \right),
$$
whereas in the case of $A_\e^{(t)}$ one has
\[
u_3(0)=\frac{\gamma_3}{\e D(k)}{\rm e}^{{\rm i}\e l_3 t},\ \ \ \ \
u_3(\e l_3)=\frac{\gamma_3}{\e D(k)}.
\]
In the cases (b), (c),
 the error between the actions of
 $\Phi_\e\bigl(A_\varepsilon^{(t)}-z\bigr)^{-1}\Phi_\e^*$ and $\bigl(\tilde A_\varepsilon^{(t)}-z\bigr)^{-1},$ up to the correcting term, is of the order
$O(\varepsilon^2),$ due to
the order $O(\varepsilon^{3/2})$ pointwise error,
the above estimate for $\gamma_3,$
and the fact that $e_1,$ $e_3$ have lengths proportional to $\e.$

We now
consider the ``correcting'' term that appears above in the analysis of the action of $\bigl(\tilde A_\e^{(t)}-z\bigr)^{-1}$ restricted to the stiff component. On the face of it, this term is $\e$-singular, however this is an artificial singularity, since this corrector is equal to the difference of resolvents of two self-adjoint operators and as such is at most of order $O(1)$. The
order $O(\varepsilon^{-1})$ singularity is due to the fact that this operator acts in the space $L^2_{\rm stiff},$ see \eqref{L2s_def}, and disappears under a unitary rescaling. The correcting term admits the form
$$
C^{(t)}_\e \binom {f_1}{f_3}:=\frac{1}{\e k^2 (l_1+l_3)}\begin{pmatrix}
1&0&0\\
0&0&1
\end{pmatrix}\tilde \gamma (z)
\begin{pmatrix}
1&0&{\rm e}^{-{\rm i}\e l_3 t}\\
0&0&0\\
{\rm e}^{{\rm i}\e l_3 t} &0&1
\end{pmatrix}
\tilde \Gamma_1^{(2)}(A_\infty-z)^{-1}\left(\begin{array}{c}f_1\\0\\ f_3\end{array}\right),
$$
and for any fixed $k\neq 0$ can be treated as a bounded linear operator on
$L^2_{\rm stiff}.$
We next show that
up to an error of order $O(\e^2)$ it is a rank-one operator
multiplied by $k^{-2}.$ The analysis leading to the equation \eqref{eq:gamma} and the similar argument pertaining to the space $L^2(e_3)$ show that
$C^{(t)}_\e$ essentially only acts on the function ${\rm e}^{-{\rm i}ty}$. As for its range, the following simple argument applies. If one seeks to compute the action of the operator $\tilde \gamma(z)$ on a vector obtained by the application of $C^{(t)}_\e$ to the vector $(f_1,0)^\top\in L_{\rm stiff}^2$, then
for the restriction to the interval $e_1$
one has the
boundary-value problem with data
\[
u_1(0)=\frac{\gamma_1}{\e k^2 (l_1+l_3)},\ \ \ \ \ \ \quad u_1(\e l_1)=\frac{\gamma_1}{\e k^2 (l_1+l_3)}{\rm e}^{-{\rm i}\e l_1 t},
\]
where $\gamma_1$ is defined by \eqref{eq:gamma} with the terms $O(\e^{5/2})\|f\|$ dropped.
Its solution is given by
$$
u_1(x)=\frac {\gamma_1}{\e k^2 (l_1+l_3)}\left\{
{\rm e}^{-{\rm i}tx} \biggl(\sin \dfrac{k\e l_1}{\sqrt{a}}\biggr)^{-1}\sin \dfrac{k(\e l_1-x)}{\sqrt{a}}
+{\rm e}^{-{\rm i}\e l_1 t}{\rm e}^{{\rm i}t(\e l_1-x)} \biggl(\sin\dfrac{k\e l_1}{\sqrt{a}}\biggr)^{-1}\sin \dfrac{kx}{\sqrt{a}}
 \right\}
$$
$$
=\frac {\gamma_1}{\e k^2 (l_1+l_3)}\left\{{\rm e}^{-{\rm i}tx}\left(\frac x{\e l_1}+O(\e^2)\right)+
{\rm e}^{-{\rm i}\e l_1 t}{\rm e}^{{\rm i}t(\e l_1-x)}\left(1-\frac x{\e l_1} +O(\e^2)\right)\right\}
=\frac {\gamma_1}{\e k^2 (l_1+l_3)}\bigl({\rm e}^{-{\rm i} t x} +O(\e^2)\bigr).
$$

For the interval $e_3$
we look at the boundary-value problem with data
\[
u_3(0)=\frac{\gamma_1}{\e k^2 (l_1+l_3)}{\rm e}^{{\rm i}\e l_3 t},\ \ \ \ \ \quad u_3(\e l_3)=\frac{\gamma_1}{\e k^2 (l_1+l_3)},
\]
whence by the same argument we get
$$
u_3(x)=\frac {\gamma_1}{\e k^2 (l_1+l_3)}\bigl({\rm e}^{{\rm i} t (\e l_3-x)} +O(\e^2)\bigr).
$$

In the situation just considered, we have $\gamma_1 =\bigl\langle f_1,e^{-{\rm i}tx}\bigr
\rangle_{L^2(e_1)}\bigl(1+O(\e^2)\bigr),$ up to an error $O(\e^{5/2})\|f\|,$ which contributes an error $O(\e^2)$ to the norm-resolvent estimate. Using the notation (\ref{Chi_capital}),
one then gets the following representation for the correcting operator:
\begin{equation}\label{eq:Cte}
C^{(t)}_\e[\cdot]=\bigl(\e k^2 (l_1+l_3)\bigr)^{-1}\bigl\langle \cdot, {\mathcal X}^{(t)}\bigr\rangle_{L^2_{\rm stiff}} {\mathcal X}^{(t)}+ O(\e^2)=\widehat C^{(t)}[\cdot]+O(\e^2),
\end{equation}
where the error estimate is understood in the sense of the operator norm in $L^2(e_1).$

Now we show that the same expression accounts for the correcting term in the situation when $C^{(t)}_\e$ is evaluated on the vector $f=(0,f_3)^\top\in L^2_{\rm stiff}$. Indeed,
up to $O(\e^{5/2})\|f\|,$ one has
\begin{multline*}
\gamma_3
={\rm e}^{-{\rm i}\e l_3 t}\bigl(\dt v\bigr|_0-{\rm e}^{{\rm i}\e l_3 t} \dt v\bigr|_{\e l_3}\bigr)
={\rm e}^{-{\rm i} \e l_3 t}\bigl\langle f_3, {\rm e}^{-{\rm i}tx}\bigr\rangle_{L^2(e_3)}\bigl(1+O(\e^2)\bigr)
=\bigl\langle f_3, {\mathcal X}^{(t)}\bigr\rangle_{L^2(e_3)}\bigl(1+O(\e^2)\bigr).
\end{multline*}
By the same argument as above we get
\eqref{eq:Cte} in the sense of the norm in $L^2(e_3).$ Summarising, the estimate \eqref{eq:Cte} holds in the sense of the norm of $L_{\rm stiff}^2.$
The required estimate \eqref{eq:intermediate} follows. \end{proof}

\begin{remark}\label{rem:high-freq} Note that the norm of $C^{(t)}$ does not depend on $\e$ when $z=k^2$ is in $S^{(t)}_{K,\rho}$. However, if one considers a transition regime from the classical setting to high frequency homogenisation, {\it i.e.,}
the situation when $z\e^\omega,$  $\omega<2,$ tends to a positive constant, its norm starts decaying as $\e\to 0$ and this term thus has no influence on the result.
\end{remark}

\section{Behaviour of the resolvents $\bigl(\tilde{A}_\varepsilon^{(t)}-z\bigr)^{-1}$ and the main result.}
\label{main_section}


The next step of our argument concerns passing to the effective, or ``homogenised'', operator $A_{\rm hom}^{(\tau)},$ which provides
the ``operator asymptotics''  for the generalised resolvent of $\tilde A_\e^{(t)}$ for all $\varepsilon, t.$
Recall that in the present paper we consider the ``finite-frequency'' case,
by assuming throughout that $z\in S^{(t)}_{K,\rho}$ (see \eqref{eq:Srho}) for some compact $K$ and $\rho>0$. First, we introduce some notation.

\begin{definition}
\label{Psi_definition}
Consider the following normalisation of the vector ${\mathcal X}^{(t)}$ defined by $(\ref{Chi_capital}):$
$$
\psi^{(t)}:=\frac 1{\sqrt{\e(l_1+l_3)}}e^{{\rm i}\e l_1 t} {\mathcal X}^{(t)},
$$
and the orthogonal projection $P_{\psi}$ in the space $L^2_{\rm stiff},$ defined by (\ref{L2s_def}), onto the vector $\psi^{(t)}.$
For convenience, in what follows we keep the same notation $\psi^{(t)}$ for the extension, by the zero element in $L^2(e_2),$ of the vector $\psi^{(t)}$ to the whole space $H=L^2(e_1)\oplus L^2(e_2)\oplus L^2(e_3)$. For all $t\in[0,2\pi\varepsilon^{-1}),$ we define a unitary operator
\[
\Psi^{(t)}: P_\psi L^2_{\rm stiff}\oplus L^2(e_2)=:H_{\rm eff}\to H_{\rm hom}:=L^2(e_2)\oplus \mathbb{C}
\]
by mapping
$
\beta\psi^{(t)}\oplus u_2 \mapsto(u_2, \beta)^\top.
$
\end{definition}

\begin{definition}
\label{Ahom}
For all values $\tau\in[0, 2\pi),$ consider an operator $A_{\rm hom}^{(\tau)}$ on the above space $H_{\rm hom},$ defined as follows. Let the domain $\text{\rm dom}\bigl(A_{\rm hom}^{(\tau)}\bigr)$ consist of all pairs $(u,\beta)$ such that $u\in W^{2,2}(e_2)$ and the quasiperiodicity condition
\begin{equation}
u(0)=w_{\rm soft} u(l_2)=\frac{\beta}{\sqrt{l_1+l_3}}
\label{quasi_cond}
\end{equation}
is satisfied.  On $\text{\rm dom}\bigl(A_{\rm hom}^{(\tau)}\bigr)$
the action of the operator is set by
$$
A_{\rm hom}^{(\tau)}\binom{u}{\beta}=
\left(\begin{array}{c}\biggl(\dfrac{1}{\rm i}\dfrac{d}{dx}+\tau\biggr)^2\\[0.9em]
-\dfrac{1}{\sqrt{l_1+l_3}}
\sum\widehat\partial^{(\tau)}_nu
\end{array}\right),\ \ \ \ \ \ \ \
\sum\widehat\partial^{(\tau)}_nu:= \partial^{(\tau)} u\bigr|_0 - w_{\rm soft}\partial^{(\tau)} u\bigr|_{l_2}.
$$
\end{definition}


As we show below, the space $H_{\rm eff}$
is ``almost invariant'' for the generalised resolvent of
$\tilde A_\e^{(t)},$ whence this resolvent can be sandwiched by projections $P_{\rm eff}$ of $H$ onto $H_{\rm eff}$
at the expense of an error of order $O(\e^2).$ Having done this, we will only consider the situation in the space $H_{\rm eff}$. The function $u$ on the space of dimension one that remains of the stiff component is then uniquely defined by its value at the vertex $V_3$, which is determined by the boundary values of
$u$ on the soft component. These boundary values are not fixed by the domain of the operator $\tilde A_\e^{(t)}$ but are nevertheless readily available by the same argument as in the proof of Theorem \ref{comparison_theorem}. Once $u_1$ and $u_3$ are determined uniquely, one can rewrite the matching conditions on the soft component that decouple it from
the stiff component. Finally, the value $u_2(0)$ uniquely determines the solution on the stiff component, up to an error of order $O(\e^2)$.

\begin{theorem}
\label{aux_proposition} The following statements hold for any $z\in S^{(t)}_{K,\rho},$ where $S^{(t)}_{K,\rho}$ defined by \eqref{eq:Srho}:

1. The norm of the difference $\bigl(\tilde A_\e^{(t)} -z\bigr)^{-1}-P_{\rm eff}\bigl(\tilde A_\e^{(t)} -z\bigr)^{-1}P_{\rm eff}$
is of the order $O(\e^2).$

2. The action of the generalised resolvent
$\bigl(\tilde A_\e^{(t)} -z\bigr)^{-1}$ on a vector $f=(f_1,f_2,f_3)^\top$ is $O(\e^2)$--close in the operator-norm sense to the vector $u=(u_1,u_2,u_3)^\top$ described as follows. The component $u_2$ is the solution of the following boundary-value problem on $e_2:$
\begin{equation}\label{eq:Rz}
\biggl(\frac{1}{\rm i}\frac d{dx}+\tau\biggr)^2 u_2-z u_2=f_2,\ \ \ \ \ \ \
u_2(0)=w_{\text{\rm soft}} u_2(l_2), \ \ \ \ \ \
\sum\widehat\partial^{(\tau)}_nu_2
=-z(l_1+l_3) u_2(0)- \sqrt{l_1+l_3}\,\bigl\langle f, \psi^{(t)}\bigr\rangle,
\end{equation}
where $\psi^{(t)}$ is extended to a vector in $H_{\rm eff}$ by zero on the soft-component space $L^2(e_2).$ For the solution $u_2$ of (\ref{eq:Rz}), the component
$u_{\rm stiff}=(u_1,0,u_3)^\top$ is determined
by
\begin{equation}\label{eq:stiff_solution}
u_{\rm stiff}= \sqrt{l_1+l_3}\,u_2(0)\psi^{(t)}-z^{-1}\bigl\langle f,\psi^{(t)}\bigr\rangle\psi^{(t)}.
\end{equation}
\end{theorem}

\begin{proof}
We use the Kre\u\i n resolvent formula, see Section \ref{sect:Krein}, that links $\bigl(\tilde A_\e^{(t)} - z\bigr)^{-1}$ to $(A_\infty - z)^{-1}.$ Notice that the Dirichlet decoupling $(A_\infty - z)^{-1}$ has the property
$$
(A_\infty - z)^{-1}\bigl[(f_1, 0, f_3)^\top\bigr]
=O(\e^2),
$$
due to the fact that the lower bound of the spectrum of its first and third components is of the order
$O\bigl(\e^{-2}\bigr)$. 
Therefore, the contribution of the Dirichlet decoupling can be ignored in the proof, and the only part of the expression for the resolvent of $\tilde A_\e^{(t)}$ that needs to be accounted for is the second term in the Kre\u\i n formula (\ref{eq:resolvent}), related to the perturbation in the boundary space from the decoupled operator.

It follows from the proof of Theorem \ref{comparison_theorem} that for all vectors  $f=(f_1, 0, 0)^\top\in H$ and  $f=(0, 0, f_3)^\top\in H$,  whose projection onto $L^2_{\rm stiff}$
is orthogonal to $\psi,$
one has $\gamma_1=O(\e^{5/2})\|f\|$ and $\gamma_3=O(\e^{5/2})\|f\|$, respectively, as $\varepsilon\to0.$
This immediately implies that restricting $\bigl(\tilde A_\e^{(t)}-z\bigr)^{-1}$ to the space $H_{\rm eff}$ results in an error of order $O(\e^2)$ in the operator norm.

In order to estimate the effect of sandwiching the resolvent between two projections onto $H_{\rm eff}$, we start by considering the vector $u:=\bigl(\tilde A_\e^{(t)}-z\bigr)^{-1}(0,f_2,0)^\top$. By an argument in the proof of Theorem
\ref{comparison_theorem}, for $u_1$ and $u_3$ one has the boundary values (up to an error of order $O(\e^{3/2})$)
$$
u_1(0)=-\frac{\gamma_2}{\sqrt{\e} D(k)} {\rm e}^{{\rm i}\e l_1 t},\ \ \ \ u_1(\e l_1)=-\frac{\gamma_2}{\sqrt{\e} D(k)},
$$
and
\begin{equation}
\ \ u_3(0)=-\frac{\gamma_2}{\sqrt{\e} D(k)}{\rm e}^{{\rm i}\e (l_1+l_3) t},\ \ \ \ u_3(\e l_3)=-\frac{\gamma_2}{\sqrt{\e} D(k)}{\rm e}^{{\rm i}\e l_1 t},
\label{u3_data}
\end{equation}
respectively. In the same way as in approximating the corrector in the proof of Theorem \ref{comparison_theorem}, we obtain
$$
u_1(x)=-\frac{\gamma_2}{\sqrt{\e} D(k)}{\rm e}^{{\rm i} t (\e l_1 -x)}\bigl(1+O(\e^2)\bigr),\ \ \ \  u_3(x)= -\frac{\gamma_2}{\sqrt{\e} D(k)} {\rm e}^{{\rm i}\e l_1 t}{\rm e}^{{\rm i}t(\e l_3 -x)}\bigl(1+O(\e^2)\bigr),
$$
whence the restriction of the function $u$ to the stiff component is given by
\[
u_1\oplus u_3=-\frac{\gamma_2}{\sqrt{\e} D(k)}{\rm e}^{{\rm i}\e l_1 t} {\mathcal X}^{(t)}\bigl(1 +O(\e^2)\bigr).
\]
The first claim of the theorem in the case of the vector $f=(0,f_2,0)^\top$ readily follows, since the error term is of order $O(\e^2)$ in $L^2_{\rm stiff}$.

Postponing to a later stage the proof
of the case when the resolvent is applied to vectors of the form $f=(f_1,0,f_3)^\top,$ we proceed with the comparison of the asymptotic formulae for the boundary values of $u_2$ and $u_3$ in order to ascertain the second claim of the theorem on the vectors $f=(0,f_2,0)^\top.$ Building up on the analysis so far, we obtain
$$
u_2(0)=-\frac{\gamma_2}{D(k)}+O(\e^2),\ \ \ \ u_3(0)=-\frac{\gamma_2}{\sqrt{\e} D(k)}{\rm e}^{{\rm i}\e (l_1+l_3) t}+O(\e^{3/2}),
$$
where the expression for $u_3(0)$ is taken from (\ref{u3_data}),
while the expression for
$u_2(0)$ was obtained in the proof of Theorem \ref{comparison_theorem}. Clearly
$$
u_3(0)=\e^{-1/2}
{\rm e}^{{\rm i}\e (l_1+l_3) t}u_2(0)+O(\e^{3/2}),
$$
and therefore, taking into account the explicit description of the domain of $\tilde A^{(t)}_\varepsilon,$ one has
\begin{equation}
\sum\widehat\partial^{(\tau)}_nu_2
=-2k^2(l_1+l_3)u_2(0)+\sqrt{\e} k^2 (l_1+l_3) e^{-{\rm i}(l_1+l_3)\tau}u_3(0)+O(\e^2)=-z(l_1+l_3)u_2(0)+O(\e^2).
\label{eq:re1}
\end{equation}
We show that dropping the $O(\e^2)$ term on the right-hand side of (\ref{eq:re1}) leads to an error of order $O(\e^2)$ in the operator-norm sense.
Indeed, as $(u_1,u_2,u_3)^\top$ is in the domain of $\tilde A_\e^{(t)}$ by construction, the component $u_2$ satisfies
\begin{equation}
\biggl(\frac{1}{\rm i}\frac d{dx}+\tau\biggr)^2u_2-z u_2=f_2, \ \ \ \ \ \ \ u_2(0)=w_{\text{soft}} u_2(l_2).
\label{u2_add}
\end{equation}
Note, that up to an $O(\e^2)$ term the problem (\ref{eq:re1})--(\ref{u2_add}) is independent of the stiff component and no longer depends on $\e$. Looking for a solution $u_2=u_0+v$, with
\[
v\in{\mathcal V}:=W^{2,2}(e_2)\cap\bigl\{
v:\,v(0)=v(l_2)=0,\ \ \dtau v\bigr|_0-w_{\text{soft}} \dtau v\bigr|_{l_2}=r_\e\bigr\},
\]
where $r_\e$ is the $O(\e^2)$ term in \eqref{eq:re1},
one arrives at the following boundary-value problem for $u_0:$
\begin{gather*}
\biggl(\frac{1}{\rm i}\frac d{dx}+\tau\biggr)^2u_0-z u_0=f_2-\biggl(\frac{1}{\rm i}\frac d{dx}+\tau\biggr)^2 v+z v,\ \ \ \ \ \
u_0(0)=w_{\text{soft}} u_0(l_2),\ \ \ \ \
\sum\widehat\partial^{(\tau)}_nu_0
=-z (l_1+l_3) u_0(0).
\end{gather*}
Whenever $z$ is outside some fixed neighbourhood of the poles of the generalised resolvent $R_{\text{soft}}(z)$ of the last boundary-value problem (it is easily seen that this set
is
defined by the dispersion relation \eqref{CC}, {\it cf.} calculation
in Sections \ref{limit_delta} and \ref{non-Bloch}), one has:
$$
u_0=R_{\text{soft}} (z)\biggl\{f_2-\biggl(\frac{1}{\rm i}\frac d{dx}+\tau\biggr)^2v+z v\biggr\},
$$
Let $\kappa$ be a constant such that $0<l_1+l_3+\kappa l_2<1/4,$ and set
$v=\alpha x(1-x/l_2){\rm e}^{-{\rm i}\kappa \tau x},$
$\alpha=r_\e\bigl(1+{\rm e}^{-{\rm i}(l_1+l_3+\kappa l_3)\tau}\bigr)^{-1}. 
$
Clearly $v\in{\mathcal V},$ and
$$
\biggl\|\biggl(\frac{1}{\rm i}\frac d{dx}+\tau\biggr)^2 v-z v\biggr\|_{L^2(e_2)}=O(\e^2)
$$
uniformly with respect to $\tau$, so that
\begin{equation}
u_0=R_{\text{soft}} (z) f_2 + O(\e^2)
\label{u0_Rsoft}
\end{equation}
in the operator-norm sense. In view of (\ref{u0_Rsoft}) and the fact that $v=O\bigl(\varepsilon^2\bigr),$ 
the estimate
$$
u_2=u_0+ O(\e^2)=R_{\text{soft}} (z) f_2 + O(\e^2)
$$
holds.
In addition, the embedding of $W^{2,2}(e_2)$ into $C(e_2)$ implies that
$$
u_2(0)=\bigl[R_{\text{soft}} (z) f_2\bigr](0)+ O(\e^2).
$$
Indeed, $R_{\text{soft}}(z)$ can again be considered as the resolvent at the point $z$ of a closed linear operator $A_z$ defined by \eqref{eq:Rz}.
Therefore away from the spectrum of $A_z,$ the operator $R_{\text{soft}}(z)$ is bounded from 
$L^2(e_2)$ to 
$\dom(A_z)$ equipped with the graph norm. As is easily seen, within the conditions of the theorem we are guaranteed to be in this situation. Denoting
\[
\tilde u =R_{\text{soft}} (z)\biggl\{\biggl(\frac{1}{\rm i}\frac d{dx}+\tau\biggr)^2 v-z v\biggr\},
\] one then has
$
\|A_z \tilde u \|^2+\|\tilde u \|^2=O(\e^4),
$
whence
$$
\biggl\|\left(\frac 1{\rm i}\frac d{dx}+\tau\right)^2 \tilde u \biggr\|^2+\|\tilde u \|^2=\biggl\|\left(\frac 1{\rm i}\frac d{dx}\right)^2{\rm e}^{{\rm i}\tau x}\tilde u \biggr\|^2+\|{\rm e}^{{\rm i}\tau x}\tilde u \|^2=O(\e^4),
$$
and $\tilde u(0)=O(\e^2)$ by the embedding theorem. Noting that $u_2=R_{\text{soft}} (z)f_2-
\tilde u +v$ and $v=O\bigl(\varepsilon^2\bigr)$ in $W^{2,2}$-norm, the claim follows.

The explicit relationship between $u_3(0)$ and $u_2(0)$ is now used to construct the solution on the stiff component. As mentioned above, this solution is fully determined by its value at the vertex $V_3:$
$$
u_{\rm stiff}=
\e^{-1/2}{\rm e}^{{\rm i}\e l_1 t}\bigl[R_{\text{soft}} (z)f_2\bigr](0){\mathcal X}^{(t)}+ O(\e^{3/2}),
$$
where the $O(\e^{3/2})$ terms leads to an order $O(\e^2)$ error in $L^2_{\rm stiff}$, as claimed.

It remains to show that both claims of the theorem hold for the resolvent applied to the right-hand side supported on the stiff component, namely $f=(f_1,0,f_3)^\top$. Since we have already shown that the resolvent $\bigl(\tilde A_\e^{(t)}-z\bigr)^{-1}$ can be restricted to the space $H_{\rm eff}$ up to an error of order $O(\e^2)$ in the operator-norm sense, we assume that $f$ is proportional to $\psi^{(t)}$. By linearity, we split the calculation into two cases, $f=(f_1,0,0)^\top$ and $f=(0,0,f_3)^\top,$ which are labelled by the index $j=1,3.$  Once again, in each of the two cases we start by reconstructing the solutions that pertain to $\tilde A_\e^{(t)}$
restricted to the stiff component. These are sums of solutions to the boundary-value problems on $[0,\varepsilon l_1],$ 
$[0,\varepsilon l_3]:$ 
\begin{gather*}
u_1(0)=-\frac {\gamma_j} {\e D(k)}\bigl(1 +O(\e^2)\bigr),\ \ \ \ \ \ u_1(\e l_1)=- \frac {\gamma_j} {\e D(k)}{\rm e}^{-{\rm i}\e l_1 t}\bigl(1 +O(\e^2)\bigr),\ \ \ \ j=1,3,\\[0.6em]
u_3(0)=-\frac {\gamma_j} {\e D(k)}{\rm e}^{{\rm i}\e l_3 t}\bigl(1 +O(\e^2)\bigr),\ \ \ \ \ u_3(\e l_3)= - \frac {\gamma_j} {\e D(k)}\bigl(1 +O(\e^2)\bigr),\ \ \ j=1,3,
\end{gather*}
and solutions to the boundary-value problems due to the corrector $C^{(t)}$. By the same asymptotic expansion as above, we get 
$$
u_1=-\frac {\gamma_j} {\e D(k)}{\rm e}^{-{\rm i}tx}\bigl(1 +O(\e^2)\bigr),\ \ \ \ \ \ u_3=-\frac {\gamma_j} {\e D(k)}{\rm e}^{{\rm i}t(\e l_3 -x)}(1 +O(\e^2)\bigr),\ \ \ j=1,3.
$$
Taking into account the contributions due to the corrector term yields
$$
u_1\oplus u_3=-\frac {\gamma_j} {\e D(k)} {\mathcal X}^{(t)}\bigl(1+O(\e^2)\bigr)-\frac 1 \e \frac {\gamma_j} {k^2(l_1+l_3)} {\mathcal X}^{(t)}, \ \ \ j=1,3,
$$
which clearly suffices to ascertain the first claim of the theorem, taking into account the estimates $\gamma_1=O(\sqrt{\e})$,  $\gamma_3=O(\sqrt{\e})$ obtained in the proof of Theorem \ref{comparison_theorem}.

In order to prove the second claim of the theorem, we proceed in the same way as above.
Using the boundary data for $u_2$, namely,
$u_2(0)=-\gamma_j\bigl(\sqrt{\e} D(k)\bigr)^{-1}{\rm e}^{-{\rm i}\e l_1 t},$ to the leading order, $j=1,3,$ we obtain for the cases
$f=f^{(1)}:=(f_1,0,0)^\top$ and $f=f^{(3)}:=(0,0,f_3)^\top,$ {\it cf.} (\ref{eq:re1}):
\begin{multline}\label{eq:dt_match_extra}
\sum\widehat\partial^{(\tau)}_nu_2= -2k^2 (l_1+l_3) u_2(0)+\sqrt{\e} k^2 (l_1+l_3) {\rm e}^{-{\rm i}(l_1+l_3)\tau}  u_3(0)
\\[0.4em]
=-2k^2 (l_1+l_3)u_2(0)+\sqrt{\e}k^2 (l_1+l_3){\rm e}^{-{\rm i}(l_1+l_3)\tau}\biggl(\frac{1}{\sqrt{\e}}
{\rm e}^{{\rm i}(l_1+l_3)\tau}u_2(0)-\frac 1\e \frac {\gamma_j} {k^2 (l_1+l_3)}{\rm e}^{{\rm i}l_3\tau}+O(\varepsilon^{3/2})\biggr)\ \ \ \ \ \ \ \ \ \ \ \  \ \ \ \\[0.3em]
= -k^2(l_1+l_3)u_2(0)-\frac 1{\sqrt{\e}}{\rm e}^{-{\rm i}l_1\tau} \gamma_j+O(\varepsilon^{2})=-z(l_1+l_3)u_2(0)-\frac 1{\sqrt{\e}}
{\rm e}^{-{\rm i}l_1\tau}\bigl\langle f^{(j)},{\mathcal X}^{(t)}\bigr\rangle+O(\varepsilon^{2}),\ \ j=1,3.
\end{multline}
Further, we discard the  $O(\e^2)$ term on the right-hand side, due to the same argument as for  $\bigl(\tilde A_\e^{(t)}-z\bigr)^{-1}\bigl[(0,f_2,0)^\top\bigr]$. The only difference in this case is that in order to reduce the problem to that for $R_{\text{soft}} (z)$, we look for the solution $u_2$ as a sum of three functions, namely $u_2=u_0+\tilde v +v$, where 
$v$ is
as above and 
\[
\tilde v\in
W^{2,2}(e_2)\cap\Bigl\{
\tilde v:\,\tilde v(0)=\tilde v(l_2)=0,\ \ \dtau \tilde v\big|_0-w_{\text{soft}} \dtau \tilde v\big|_{l_2}=-\frac 1{\sqrt{\e}}
{\rm e}^{-{\rm i}l_1\tau}\bigl\langle f^{(j)},{\mathcal X}^{(t)}\bigr\rangle\Bigr\},\ \ \ j=1,3,
\]
is constructed in 
the same way as $v$. The function
\[
\tilde f_2:=-\biggl(\frac 1{\rm i}\frac d{dx}+\tau\biggr)^2 \tilde v+z \tilde v
\]
then takes the place of the function $f_2$ in the corresponding construction for $u_0$ in the case $f=(0,f_2,0)^\top,$  allowing  
to drop an error term of order $O(\e^2)$ in $u_2,$ by an application of the same embedding theorem. Finally, the function $u_2=R_{\text{soft}} (z)\tilde f_2+\tilde v$  solves the boundary-value problem \eqref{eq:Rz}, since
in terms of the 
function $\psi^{(t)}$ the boundary condition 
\eqref{eq:dt_match_extra}  reads
\begin{equation*}
\sum\widehat\partial^{(\tau)}_nu_2
= -k^2(l_1+l_3)u_2(0) - \sqrt{l_1+l_3}\,\bigl\langle f^{(j)}, \psi^{(t)}\bigr\rangle,\ \ \ \ j=1,3.
\end{equation*}
This completes the proof.
\end{proof}


\begin{theorem}
\label{second_claim}
For a given compact $K$ and $\rho>0,$ let $S^{(t)}_{K,\rho}$ be defined by \eqref{eq:Srho}. The asymptotic formula
\begin{equation}
\label{eq:claim_corrected}
\Psi^{(t)}\bigl[\bigl(\tilde A_\e^{(t)}-z\bigr)^{-1}+C^{(t)}\bigr]\bigl(\Psi^{(t)}\bigr)^*=\bigl(A_{\rm hom}^{(\tau)}-z\bigr)^{-1}+O(\varepsilon^2),\ \ \ \ \ \varepsilon\to0,
\end{equation}
holds, with the error understood in the sense of the operator norm, uniformly with respect to $z\in S^{(t)}_{K,\rho}.$ In the
formula (\ref{eq:claim_corrected}), the unitary operator $\Psi^{(t)}$ is given by Definition \ref{Psi_definition}; the operator $C^{(t)}$ is given by \eqref{Ct_def1},
or equivalently
$
C^{(t)}[\cdot]=
z^{-1}\bigl\langle \cdot, \psi^{(t)}\bigr\rangle\psi^{(t)},
$
where $\psi^{(t)}$ is extended to a vector in $H_{\rm eff}$ by zero on the soft-component space $L^2(e_2)$.

\end{theorem}

\begin{proof}

I. We first verify the claimed identity on vectors $(f, 0)^\top.$ Notice that $\bigl(\Psi^{(t)}\bigr)^*\bigl[(f, 0)^\top\bigr]=(0,f,0)^\top$, which is the case considered in Theorem \ref{aux_proposition} with $f_2=f,$ where we show that the following representation for the action of the resolvent on the left-hand side holds, in the limit as $\e\to0:$
\begin{gather*}
 \biggl(\frac 1{\rm i}\frac{d}{dx}+\tau\biggr)^2 u_2-z u_2=f,\ \ \ 
u_2(0)=
w_{\rm soft}
u_2(l_2),\ \ \ 
\sum\widehat\partial^{(\tau)}_nu_2
=-z(l_1+l_3) u_2(0),
\ \ \ u_{\rm stiff}=\sqrt{l_1+l_3}\,u_2(0)\psi^{(t)}.
\end{gather*}
Evaluating $\Psi^{(t)}$ on the vector $u_1\oplus u_2 \oplus u_3\equiv u_2 \oplus u_{\rm stiff}$ yields
$\Psi^{(t)}(u_2\oplus u_{\rm stiff})=(u_2, \beta)^\top,$
$\beta = \sqrt{l_1+l_3}\,u_2(0).$
On the other hand, for the action $\bigl(A_{\rm hom}^{(\tau)}-z\bigr)^{-1}\bigl[(f, 0)^\top\bigr]=:(u, \beta_u)^\top$  of the right-hand side of
(\ref{eq:claim_corrected}) on the same vector, one has:
\begin{gather*}
\biggl(\frac 1{\rm i}\frac{d}{dx}+\tau\biggr)^2 u-z u=f,\ \ \ \ \ \
u(0)=w_{\rm soft}
u(l_2)=\frac{\beta_u}{\sqrt{l_1+l_3}},\ \ \ \ \ \ \
\sum\widehat\partial^{(\tau)}_nu
=-z(l_1+l_3) u(0),
\end{gather*}
which is clearly the same as for the left-hand side of \eqref{eq:claim_corrected}. This completes the first part of the proof.

II. By linearity, it suffices to verify the validity of our claim on vectors of the form $(0, \beta_f)^\top.$
We have $\bigl(\Psi^{(t)}\bigr)^*\bigl[(0, \beta_f)^\top\bigr]=(f_1,0,f_3)^\top$ so that $f_1\oplus f_3=\beta_f\psi^{(t)}$. Noting that the inner product in the last term of \eqref{eq:Rz} equals $\beta_f$ and using Theorem \ref{aux_proposition} again, we obtain, for the action of the left-hand side of \eqref{eq:claim_corrected}:
\begin{gather*}
\biggl(\frac 1{\rm i} \frac{d}{dx}+\tau\biggr)^2 u_2-zu_2=0,
\\[0.4em]
u_2(0)=w_{\rm soft}
u_2(l_2),\ \ \ \ \
\sum\widehat\partial^{(\tau)}_nu_2=-z(l_1+l_3) u_2(0)- \sqrt{l_1+l_3}\,\beta_f,
\ \ \ \ \ u_{\rm stiff}= \sqrt{l_1+l_3}\,u_2(0)\psi^{(t)}.
\end{gather*}
Once again, one has
$\Psi^{(t)}(u_2\oplus u_{\rm stiff})=(u_2, \beta)^\top,$
$\beta = \sqrt{l_1+l_3}\,u_2(0).$
We consider the result of applying the resolvent $\bigl(A_{\rm hom}^{(\tau)}-z\bigr)^{-1}$ to the vector $(0, \beta_f)^\top:$
\begin{gather*}
\biggl(\frac 1{\rm i} \frac{d}{dx}+\tau\biggr)^2 u-zu=0,\ \ \ \ \ \ \
u(0)=w_{\rm soft}
u(l_2)=\frac{\beta_u}{\sqrt{l_1+l_3}},\ \ \ \ \ \ \
\sum\widehat\partial^{(\tau)}_nu
=-z(l_1+l_3) u(0)-\sqrt{l_1+l_3}\,\beta_f ,
\end{gather*}
and note that $u=u_2,$ $\beta_u=\beta,$ which concludes the proof of the claim.
\end{proof}

\begin{corollary}
\label{main_result_statement}
For
$z\in S^{(t)}_{K,\rho},$
there exists a constant $C>0$ independent of $t$ such that
\begin{equation}
\bigl\Vert\bigl(A_\e^{(t)}-z\bigr)^{-1}-\widetilde{P}_\psi\Phi_\varepsilon^*\bigl(\Psi^{(t)}\bigr)^*\bigl(A_{\rm hom}^{(\tau)}-z\bigr)^{-1}\Psi^{(t)}\Phi_\e\widetilde{P}_\psi\bigr\Vert\le C\varepsilon^2,\ \ \ \ \ \ \ \ \tau=\varepsilon t,
\end{equation}
for all $\e\in(0,1]$ and $t\in[0,2\pi\varepsilon^{-1}).$ Here $\widetilde{P}_\psi:=P_\psi\oplus I_2,$ where $I_2$ is the identity operator on $L^2(e_2).$
\end{corollary}




\begin{remark}
The function $u$ in the eigenvalue problem
$$
A_{\rm hom}^{(\tau)}\binom{u}{\beta}=z\binom{u}{\beta}
$$
is the solution to
\begin{equation}\label{limit_diff_exp}
\biggl(\frac 1{\rm i}\frac {d}{dx}+\tau\biggr)^2u=z u,
\end{equation}
\begin{equation}\label{limit_diff_exp_1}
u(0)=w_{\rm soft} u(l_2),\ \ \ \ \ \ \ \
\sum\widehat\partial^{(\tau)}_nu
=-z (l_1+l_3) u(0),
\end{equation}
where the last condition follows from the equation on the second
components. This coincides with the problem for the ``eigenvectors'' of the energy-dependent boundary-value problem obtained as a Datta -- Das Sarma modification of the
problem considered in Section \ref{sect:motivation}. Moreover,
the two can be shown to be  isospectral (and hence isospectral with the limiting operator $A_{\rm hom}^{(\tau)}$).

The argument leading to Theorem \ref{second_claim} further implies that the operator $A_{\rm hom}^{(\tau)}$, which serves as the norm-resolvent limit of the operator family $A_\varepsilon^{(t)}$, is an out-of-space extension of the related minimal operator (see Section \ref{sect:Krein}, equation \eqref{eq:out-of-space}) corresponding to the generalised resolvent of the spectral boundary-value problem \eqref{limit_diff_exp}--\eqref{limit_diff_exp_1}.

\label{remark_label}
\end{remark}

\section{Transformation to a Kronig-Penney model of $\delta'$-type: Bloch spectrum}
\label{Kronig_Penney}

Now we turn our attention to the question of unitary transformation of the direct integral of homogenised fibre operators $A_{\text{hom}}^{(\tau)}$ into the operator in the original Hilbert space $L^2(\mathbb R)$. We claim that
$A_{\rm hom}^{(\tau)}$ can be
transformed to an operator with non-trivial $\delta'$-type coupling condition (with an \emph{energy-independent} domain description). This transformation, which will be calculated below explicitly on eigenvectors of either operator, involves a change in the magnetic potential. Followed by the application of the inverse Gelfand transform, see Section \ref{Gelfand}, this results in a periodic operator on the real line ${\mathbb R}.$ We recall that  $\tau=\e t$, so that $\tau\in[0,2\pi)$.

\subsection{Limit fibre representation of $\delta$-type: Bloch spectrum}
\label{limit_delta}

We first calculate the eigenfunctions of the self-adjoint operator $A_{\rm hom}^{(\tau)}$. Its spectrum consists of two parts: the $\tau$-dependent spectrum (``Bloch spectrum'') described by the corresponding dispersion relation
and, possibly, the ``non-Bloch'' part of the spectrum, which is not described by the same and which we
calculate explicitly in Section \ref{non-Bloch} after discussing the Bloch spectrum. In order to compute the eigenfunction corresponding to any of the energies described by the dispersion relation, one must consider solutions to the differential equation
\begin{equation}
\left(\frac {1}{\rm i}\frac {d}{dx} +\tau\right)^2 u = z u\ \  \text { on }\  e_2.
\label{above_equation}
\end{equation}
For the Bloch spectrum, one has the boundary-value problem
\begin{equation}\label{eq:delta_conditions}
u(0)=w_{\rm soft}
u(l_2),\ \ \ \ \ \
\sum\widehat\partial^{(\tau)}_nu
=-z(l_1+l_3) u(0),
\end{equation}
under the additional condition $\sin(k l_2)\not=0.$
The solution $u=u(\alpha,\beta; \cdot)$ of (\ref{above_equation})
 subject to the conditions $u(0)=\alpha, u(l_2)=\beta,$ is then given by
$$
u(\alpha, \beta; x)= \alpha {\rm e}^{-{\rm i}\tau x}\frac {\sin k(l_2-x)}{\sin{k l_2}}+
\beta {\rm e}^{{\rm i}\tau (l_2-x)} \frac {\sin kx}{\sin k l_2},\ \ \ \ \ x\in e_2.
$$
The boundary condition involving normal derivatives then yields ({\it cf.} \eqref{CC}) the dispersion relation
\begin{equation}
2 \cot k l_2-2 \cos{\tau} \csc k l_2 = k (l_1+l_3).
\label{dispersion_rel}
\end{equation}
Therefore, for the eigenvectors $\bar{u}$ of the operator $A_{\rm hom}^{(\tau)}$ on the space $H_{\rm hom}$ one has
\begin{equation}
\bar{u}(k)=\left(\begin{array}{c}u(1, {\rm e}^{{\rm i}(l_1+l_3)\tau}; \cdot)\\[0.3em] \sqrt{l_1+l_3}\end{array}\right),
\label{ubar_vector}
\end{equation}
where
$$
u\bigl(1, {\rm e}^{{\rm i}(l_1+l_3)\tau}; x\bigr)= {\rm e}^{-{\rm i}\tau x}\left(\frac {\sin k(l_2-x)}{\sin{k l_2}}+
 {\rm e}^{{\rm i}\tau} \frac {\sin kx}{\sin k l_2}\right),\ \ \ \ \ x\in e_2,
$$
subject to the dispersion relation \eqref{dispersion_rel} holding so that $k^2$ is in the spectrum.
A straightforward integration then yields:
$$
\bigl\|\bar{u}(k)\bigr\|^2_{H_{\rm hom}}=\frac{l_1+l_3}2 +\frac {l_2}{(\sin k l_2)^2}\left( 1-\cos\tau \cos kl_2\right).
$$
The first component in (\ref{ubar_vector})
is rewritten using
(\ref{dispersion_rel}):
\begin{multline}
u\bigl(1, {\rm e}^{{\rm i}(l_1+l_3)\tau}; x\bigr)={\rm e}^{-{\rm i}\tau x}\left(\frac {\sin k(l_2-x)}{\sin{k l_2}}+
{\rm e}^{{\rm i}\tau} \frac {\sin kx}{\sin k l_2}\right)={\rm e}^{-{\rm i}\tau x} \left( \cos kx -\frac{\cos kl_2}{\sin{k l_2}} \sin kx + {\rm e}^{{\rm i}\tau} \frac {\sin kx}{\sin kl_2}\right)\\[0.6em]
={\rm e}^{-{\rm i}\tau x} \left( \cos kx + \left[-k \frac{l_1+l_3}2 + {\rm i}\frac {\sin \tau}{\sin k l_2}\right]\sin kx\right),\ \ \ \ \ x\in e_2.
\label{uform}
\end{multline}

\subsection{Limit fibre representation of $\delta'$-type: Bloch spectrum}
\label{deltaprimesection}

Consider the operator $A'_{\rm hom}=A'_{\rm hom}(\tau')$ in the space $L^2(e_2)$
defined by the same differential expression as $A_{\rm hom}^{(\tau)},$ with the parameter $\tau$ replaced by
$\tau':$
\[
\biggl(\frac{1}{\rm i}\frac {d}{dx}+\tau'\biggr)^2,
\]
on the domain described by the conditions
\begin{eqnarray}
u(0)+{\rm e}^{-{\rm i}(l_1+l_3)\tau'} u(l_2)=(l_1+l_3) {\partial^{(\tau')}} u\bigr|_0,\label{u_condition}\\[0.4em]
{\partial^{(\tau')}} u\bigl|_0=-{\rm e}^{-{\rm i}(l_1+l_3)\tau'}{\partial^{(\tau')}} u\bigr|_{l_2}.\label{uprime_condition}
\end{eqnarray}
Note that the above conditions are written equivalently as
\begin{equation}
u(0)+{\rm e}^{-{\rm i}(l_1+l_3)\tau'} u(l_2)=(l_1+l_3)\widehat{\partial}_n^{(\tau')}u\bigr|_0,\ \ \ \ \ \ \
\widehat{\partial}_n^{(\tau')} u\bigl|_0=\widehat{\partial}_n^{(\tau')} u\bigr|_{l_2},
\label{khren'}
\end{equation}
by passing over to the corresponding Datta -- Das Sarma modification, \emph{i.e.}, by associating the weight ${\rm e}^{-{\rm i}(l_1+l_3)\tau'}$ with the right endpoint of the interval
$e_2.$
The operator $A'_{\rm hom}$ is a self-adjoint extension of $\delta'$ type, {\it i.e.} it can be formally written as a
$\delta'$-type perturbation of a second-order differential operator, see, \emph{e.g.} \cite{Exner2,Kuchment2}. The coupling constant corresponding to this $\delta'$-type matching condition is $l_1+l_3$.

The boundary triple for the operator $A'_{\rm hom}(\tau')$ can be chosen \cite{Yorzh3} so that
 the boundary space is $\mathcal H=\mathbb C$ and the boundary operators are
\begin{equation}\label{eq:delta-prime-triple}
\Gamma_0'u = \widehat{\partial}_n^{(\tau')} u\bigr|_V, \quad \Gamma_1'u = - \sum_{x\in V} w(x)u(x),
\end{equation}
where both endpoints of $e_2$
are identified with each other, so that $e_2$ forms a loop attached at the vertex $V.$ The left and right endpoints of this loop are assigned the Datta -- Das Sarma weights 1 and
${\rm e}^{-{\rm i}(l_1+l_3)\tau'},$ respectively.
 The parameterising matrix (see Definition \ref{Def_BoundTrip}) is the scalar $B'=-(l_1+l_3)$.

The spectrum of $A'_{\rm hom}(\tau')$ is discrete and consists of Bloch-type eigenvalues and, possibly, eigenvalues of non-Bloch type. With respect to the  boundary triple  introduced above  these parts of the spectrum also correspond to the spectrum that is ``visible'' to the $M$-matrix of the maximal operator and the one which is ``invisible'' to it (as eigenvalues of the corresponding minimal operator which is then non-simple).


The Bloch spectrum
is characterised in the following way. At a given $k$, consider the solution to the spectral equation with the boundary data
\[
\widehat{\partial}_n^{(\tau')} u\big|_0=\widehat{\partial}_n^{(\tau')} u\big|_{l_2}=k.
\]
The corresponding solution is given by
$$
v(x;k)={\rm e}^{-{\rm i}\tau'x} \left(  \frac{\cos k(l_2-x)}{\sin k l_2}+{\rm e}^{{\rm i}\tau'} \frac {\cos kx}{\sin k l_2} \right), \ \ \ \ \ x\in e_2.
$$
Clearly, this is an eigenfunction of the operator $A'_{\rm hom}(\tau')$ provided that
\begin{equation}
2 \cot k l_2+2 \cos{\tau'} \csc k l_2 = k (l_1+l_3).
\label{tau_prime_dispersion}
\end{equation}
Note that if $\tau'=\tau+\pi\ {\rm (mod}\ 2\pi{\rm )}$, the dispersion relation for $A'_{\rm hom}(\tau')$ at $k$ is identical to the one for $A_{\rm hom}^{(\tau)}$, see \eqref{dispersion_rel}, at the same point $k$, and hence their Bloch spectra coincide. Combining the dispersion relation (\ref{tau_prime_dispersion}) and the expression for $v(x;k)$ yields
\begin{equation}
v(x;k)={\rm e}^{-{\rm i}\tau' x} \left(
\sin kx + \left[k\frac{l_1+l_3}2+{\rm i}\frac{\sin \tau'}{\sin k l_2}\right] \cos kx
\right), \ \ \ \ \ x\in e_2.
\label{v_exp}
\end{equation}
It is checked that
$$
\bigl\|v(\cdot,k)\bigr\|^2= \frac{l_1+l_3}2 +\frac {l_2}{(\sin k l_2)^2}\left( 1+\cos\tau'\cos kl_2\right),
$$
{\it i.e.} the norms of $\bar{u}(k)$ and $v(\cdot;k)$ coincide for $\tau$ and $\tau',$ respectively, when
$\tau'=\tau+\pi\,{\rm (mod}\ 2\pi{\rm )}$. Finally, substituting $\tau'=\tau+\pi\ {\rm (mod\ }2\pi{\rm )}$ into (\ref{v_exp}) yields the following formula for
$v(x;k)$ in terms of the parameter $\tau:$
\begin{equation}
v(x;k)={\rm e}^{-{\rm i} \pi x} {\rm e}^{-{\rm i}\tau x} \left(
\sin kx + \left[k\frac{l_1+l_3}2-{\rm i}\frac{\sin \tau}{\sin k l_2}\right] \cos kx
\right), \ \ \ \ \ x\in e_2,
\label{vform}
\end{equation}
which we compare below with the first component of the eigenvector $\bar{u}(k).$

\subsection{Non-Bloch spectrum in the $\delta$- and $\delta'$-type cases
}\label{non-Bloch}
As far as the non-Bloch spectrum is concerned,
for the operator $A_{\rm hom}^{(\tau)}$ one has to solve the spectral equation (\ref{above_equation})
when $\sin(k l_2)=0$ subject to the boundary conditions \eqref{eq:delta_conditions}.
While a general solution of (\ref{above_equation}) has the form
$
u=A {\rm e}^{-{\rm i}\tau x}{\rm e}^{{\rm i} k x}+ B{\rm e}^{-{\rm i} \tau x}{\rm e}^{-{\rm i} k x},
$
the conditions (\ref{eq:delta_conditions}) are shown to imply that $\cos k l_2 ={\rm e}^{{\rm i}\tau}$ and the solution sought admits the form
$u=u(0){\rm e}^{-{\rm i}\tau x}{\rm e}^{{\rm i} k x}+C{\rm e}^{-{\rm i} \tau x}\sin kx $ with an arbitrary $C\in{\mathbb C}.$
This leads to
the eigenvector ${\rm e}^{-{\rm i}\tau x} \sin kx$ at the values $\tau = 0$, $\tau =\pi,$
where $k=\pi m/l_2$ for an even \emph{non-zero} (for $\tau = 0$) or odd (for $\tau =\pi$) value of $m,$ and to the eigenvector ${\rm e}^{-{\rm i}\tau x}\equiv 1$ for $\tau=0,$ $m=0.$

The non-Bloch spectrum of the
operator $A'_{\rm hom}$ can be treated in a similar way, which allows for a simplification since, as argued in Section \ref{deltaprimesection}, it is the set of eigenvalues of the minimal (symmetric) operator $A_{\min},$ the domain of which is uniquely defined by the  boundary triple \eqref{eq:delta-prime-triple} via conditions $\Gamma_0 u=\Gamma_1 u =0$ (see also \cite{Yorzh1} for further details). These  eigenvectors  satisfy the spectral equation and the boundary conditions that determine the domain of the minimal operator:
$$
\dtau u\bigr|_0=\dtau u\bigr|_{l_2}=0,\ \  \quad u(0) +{\rm e}^{-{\rm i}(l_1+l_3)\tau} u(l_2)
=0.
$$
The general solution is the same as above, while the boundary conditions yield
$
A=B,$ $\sin kl_2=0,$ $\cos k l_2= -{\rm e}^{{\rm i}\tau}.
$
This system has a solution for $\tau =0$ and $\tau=\pi,$ where the associated eigenfunction is given by
${\rm e}^{-{\rm i}\tau x} \cos kx,$
$k=\pi m/l_2$ for an odd or even $m,$ respectively.
If follows immediately that the operator $A_{\rm hom}^{(\tau)}$ at $\tau =0,$ $\tau=\pi$ has the same non-Bloch spectrum as $A'_{\rm hom}(\tau')$ at $\tau'=\pi,$ $\tau'=0,$ respectively.




\subsection{Unitary equivalence of $A_{\rm hom}^{(\tau)}$ and $A'_{\rm hom}(\tau')$ and the whole-line form of the limit model}
\label{Sect:Kronig_Penney}

Since $A_{\rm hom}^{(\tau)}$ and $A'_{\rm hom}(\tau')$ are self-adjoint operators with purely discrete spectra in
$H_{\rm hom}$ and $L^2(e_2)$, respectively,
for each $\tau$ and $\tau'$their
eigenfunctions form orthogonal bases in these spaces.
It follows from the above analysis that for each $\tau$ the operator $A_{\rm hom}^{(\tau)}$ is unitarily equivalent to
$A'_{\rm hom}(\tau')$,
$\tau'=\tau+\pi\,{\rm (mod}\ 2\pi{\rm )}.$
The corresponding unitary transformation is described by mapping, for each value of $k,$ the eigenfunctions of
$A_{\rm hom}^{(\tau)}$ with the first component (\ref{uform}) to the eigenfunctions (\ref{vform}) of $A'_{\rm hom}(\tau'),$
as well as the respective eigenfunctions of the non-Bloch spectra (see Section \ref{non-Bloch}).
Notice that formally this is equivalent to the simultaneous substitution of $\cos kx $ by $\sin kx$ and $\sin kx$ by $-\cos kx$ in (\ref{uform}).

Finally, we rewrite the eigenvalue problems for the operators $A_{\rm hom}'(\tau')$ in a form convenient for the application of the inverse
Gelfand transform,  see Section \ref{Gelfand}. This is followed by the description of an operator in $L^2({\mathbb R})$ of the Kronig-Penney type, whose image under the Gelfand transform is given by the family  $A_{\rm hom}'(\tau'),$ $\tau'\in[0,2\pi).$
To this end, introduce a new function $\tilde{u}$ in (\ref{u_condition})--(\ref{uprime_condition}) by the formula
\begin{equation}
\tilde{u}(y)={\rm e}^{{\rm i}l_2 y\tau'}u(l_2y),\ \ \ \ \ y\in[0,1],
\label{u_form}
\end{equation}
results in the following conditions for $\tilde{u}:$
\begin{equation*}
\tilde{u}(0)-{\rm e}^{-{\rm i}\tilde{\tau}}\tilde{u}(1)=\frac{l_1+l_3}{l_2}\tilde{u}'(0),\ \ \ \ \ \ \
\tilde{u}'(1)={\rm e}^{{\rm i}\tilde{\tau}}\tilde{u}'(0),
\end{equation*}
where
$\tilde{\tau}=\tau'+\pi\, {\rm (mod}\ 2\pi{\rm )},$
which returns the original value of the fibre parameter $\tau\, {\rm (mod}\ 2\pi{\rm )}.$
Now, considering
\begin{equation}
v(y)={\rm e}^{{-\rm i}\tilde{\tau} y}\tilde{u}(y), \ \ \ \ \ \ y\in[0,1],
\label{v_form}
\end{equation}
results in
\begin{gather}
v(1)-v(0)=-\frac{l_1+l_3}{l_2}\biggl(\frac{d}{dy}+{\rm i}\tilde{\tau}\biggr)v\biggr\vert_0,\ \ \ \ \ \ \ \ \ \
\biggl(\frac{d}{dy}+{\rm i}\tilde{\tau}\biggr)v\biggr\vert_1=\biggl(\frac{d}{dy}+{\rm i}\tilde{\tau}\biggr)v\biggr\vert_0.\label{second_cond}
\end{gather}
The
differential expression on the left-hand side of (\ref{limit_diff_exp}) takes the following form in terms of the function $v:$
\[
\frac{1}{l_2^2}\biggl(\frac{1}{\rm i}\frac {d}{dy}+\tau'+\pi\biggr)^2v=
\frac{1}{l_2^2}\biggl(\frac{1}{\rm i}\frac {d}{dy}+\tilde{\tau}\biggr)^2v,
\]
Hence, the limit Kronig-Penney model is given in each fibre $\tilde{\tau}\in[0, 2\pi)$ by the spectral equation
\begin{equation}
\frac{1}{l_2^2}\biggl(\frac{1}{\rm i}\frac {d}{dy}+\tilde{\tau}\biggr)^2v=zv,
\label{tau_spectrum}
\end{equation}
subject to the conditions
(\ref{second_cond}). Finally, the inverse Gelfand transform (\ref{inverse_scaled_Gelfand}) results in the following spectral problem on
${\mathbb R}$ for $U$ such that $\hat{U}=v,$ {\it cf.} (\ref{Gelfand_formula}):
\begin{equation}
-l_2^{-2}U''=zU,\ \ \ \ \ \ U'\in C(\mathbb R),
\ \ \ \ \ \forall n\in\mathbb Z\ \ \ U\in C[n,n+1], \ \ \ U(n+0)-U(n-0)=
l_2^{-1}(l_1+l_3)U'(n),
\label{deltaprime_eigen}
\end{equation}
where $l_1+l_2+l_3=1.$ Notice that in the case when $l_2=1$ ({\it i.e.} the stiff component is absent) we obtain the
usual operator $-d^2/{dx^2}$ on ${\mathbb R}.$ The spectral problem (\ref{deltaprime_eigen}) describes (generalised)
eigenfunctions of the operator ${\mathcal A}_{\rm hom}'$ in $L^2({\mathbb R})$ given by the diferrential expression
$-l_2^{-2}d^2/{dx^2}$ on
\[
{\rm dom}({\mathcal A}_{\rm hom}')=\bigl\{U: \forall n\in\mathbb Z\ \ U\in W^{2,2}(n, n+1),\ \ U'\in C(\mathbb R),\ \forall n\in\mathbb Z\ \ U(n+0)-U(n-0)=
l_2^{-1}(l_1+l_3)U'(n)\bigr\}.
\]

\section{Relation to earlier results}
\label{discussion}

1. Our approach via the theory of boundary triples and Krein formula offers a strategy to obtain operator-norm resolvent convergence estimates for the setting of
\cite{Exner},
\cite{KuchmentZeng}, \cite{KuchmentZeng2004},
who discuss the behaviour of the spectra of operator sequences associated with
``shrinking'' domains as in Fig.\,\ref{domain_figure}.
\begin{figure}
\begin{center}
\includegraphics[width=.8\textwidth]{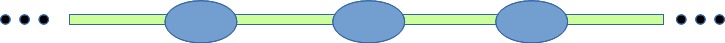}
\end{center}
\caption{{\scshape The ``fattened graph''} of the earlier works on spectral convergence of the Laplace operator on thin domains with Neumann boundary conditions.}
\label{domain_figure}
\end{figure}
Here the rate of shrinking of the green ``edge'' parts
is assumed to be related to the rate of shrinking of the  blue ``vertex'' parts via
\begin{equation}
\frac{\text{vol}(V^\varepsilon_{\rm vertex})}{\text{vol}(V^\varepsilon_{\rm edge})}\to\alpha\ge0, \ \ \ \varepsilon\to0.
\label{shrinking_rate}
\end{equation}
It is shown in the above works (for the case $\alpha=0$ in \cite{KuchmentZeng})
that the spectra of the corresponding Laplacian operators with Neumann boundary conditions
converge to the spectrum of an operator on a
one-dimensional lattice obtained as the limit of the domain in Fig.\,\ref{domain_figure} as $\varepsilon\to0.$
Our operator $A_{\rm hom}^{(\tau)},$ see Definition \ref{Ahom}, coincides with the limit operator
in \cite{Exner}, \cite{KuchmentZeng2004}.
The ``weight'' $l_1+l_3$ in our analysis
plays the r\^{o}le of the constant $\alpha$ in (\ref{shrinking_rate}), see {\it e.g.} (\ref{limit_diff_exp_1}).

In view of our results, it is intriguing to consider the one-dimensional high-contrast
problem (\ref{orig_problem})--(\ref{high_contrast_case}) as
an equivalent (in the resolvent sense) of Neumann Laplacians defined on a two-dimensional domain shrinking to an infinite chain graph, under the assumption (\ref{shrinking_rate}) with $\alpha\neq0.$
This should allow for the treatment of the  homogenisation problem in terms of resonant properties of thin structures,
thereby relating properties that are due to high contrast to properly chosen ``sizes'' of resonators located at the chain vertices. 
It would be instructive to compare such results with \cite{Zhikov_Pastukhova_singular_structures},
where $\alpha=0$ and thus the effective operator is the Laplacian on a
periodic graph with standard Kirchhoff conditions at the vertices, fully in line with the results
of \cite{Exner}, \cite{KuchmentZeng}. Notably, a resonance scattering theory approach
to the treatment of effective transmitting properties of thin graph-like structures has been developed in
\cite{pavlov2, pavlov1, pavlov3} and references therein, whose results, in our view, pave the way for
yet another promising approach to the treatment of homogenisation problems with high contrast.


\

\noindent 2. To the best of our knowledge, the fact that the limiting operator of \cite{Exner}, \cite{KuchmentZeng2004} is unitarily equivalent to a Laplacian with a non-trivial $\delta'$-type perturbation supported on an infinite one-dimensional lattice is observed in the present paper for the first time.  Building upon the results of \cite{Exner}, \cite{KuchmentZeng2004}
in the special case of infinite chain graphs, this further reveals the meaning of $\delta'$-type coupling conditions in quantum graphs,
which has attracted considerable attention during the past decade. We conjecture that the same effect occurs in the general case of periodic metric graphs, which will be discussed in a forthcoming publication.

\

\noindent 3.  Our main result, Corollary \ref{main_result_statement}, describes the asymptotic behaviour of  the problem (\ref{generic_eq}), (\ref{generic_high_contrast}) in classical operator-theoretic terms, and is similar in this to the
work \cite{CC}, where resolvent estimates of order $O(\varepsilon)$ are obtained in the multi-dimensional case $d\ge 2$ under the assumption ${\rm dist} (Q_0, \partial Q)>0,$ see (\ref{generic_high_contrast}).
We do not rely on the techniques based on two-scale convergence, which have otherwise been used in the analysis of high-contrast problems, see \cite{Zhikov2000}, \cite{BouchitteFelbacq}, \cite{KamSm2}. Our approach provides asymptotic estimates that are both norm-sharp and $\varepsilon$-order sharp, and is free from restrictions on the geometry of the composite (except for minimal smoothness assumptions on the interfaces), which in our view shows the potential of operator-theoretic techniques in the study of ``non-classical'' periodic media.

\

\noindent 4. In  the work \cite{BirmanSuslina} the effective model (\ref{effective}) was derived by an asymptotic analysis of the fibre decomposition of the resolvents (\ref{resolvents_scaling}) and a fundamental notion of spectral germ was introduced, as an operator-theoretic tool for the analysis of the
``threshold behaviour" of (\ref{resolvents_scaling}) when the parameter $\varepsilon^2z<0$ approaches the spectrum at zero. The approach of
\cite{BirmanSuslina}
applies to operators that can be defined in terms of pencils
of the form $(X_0+tX_1)^*(X_0+tX_1),$ $t\in[0,1),$ ${\rm ker} X_0\neq\{0\},$ under some additional technical assumptions on $X_0,$ $X_1.$
However, a key requirement of this approach concerning the behaviour of the pencil, namely that the number of its eigenvalues in a sufficiently small neighbourhood of zero is finite,
is not satisfied in the case of the pencil (\ref{fibre_equation}), (\ref{generic_high_contrast}), where the r\^{o}le to $t$ is played by $\vert\varkappa\vert,$ see a related discussion in Section \ref{intro_section}.  From this perspective, one of the main results of our analysis is the development of a generalised notion of spectral germ for high-contrast periodic problems. While such an object would seem to have to involve an infinite set of data, due to a growing  (as $\varepsilon\to0$) set of eigenvalues of the pencil in any given neighbourhood of zero, it is remarkable that our limit model is a simple quantum
graph with non-trivial, dipole-type interface conditions (\ref{khren'}).


\

\noindent 5.  All the ingredients of our approach to high-contrast problems of the kind
(\ref{generic_eq}), (\ref{generic_high_contrast}) are either already formulated in an abstract operator-theoretic form or can be reformulated in such a form, despite the fact that the proofs of Theorems \ref{comparison_theorem}, \ref{second_claim}
involve a list of explicit one-dimensional calculations.
In particular, in the multi-dimensional case $d\ge2$ we expect Fig.\,\ref{fig:mod} to be relevant, illustrating the related modification procedure in terms of its one-dimensional sections.
It is for this reason that we believe in the strong potential of our approach for the treatment of PDE settings. This will be realised under an appropriate modification of the classical boundary triple setup, whose abstract version \cite{DM} is not directly applicable to the PDE case.
At the same time, a suitable generalisation is readily available for one-dimensional graphs that are periodic in several directions, which we shall also address elsewhere.

\subsection*{Acknowledgements} This work was carried out under the financial support of
the Engineering and Physical Sciences Research Council (Grant EP/L018802/1 ``Mathematical foundations of metamaterials: homogenisation, dissipation and operator theory''). AVK is grateful to the University of Bath for hospitality during his research visit in 2015,
when the main body of research leading to the present publication was carried out. The work of AVK was also partially supported by a grant of the Ukrainian Ministry for Education. We are specially indebted to Dr Shane Cooper for his careful reading of the manuscript.

\end{document}